\documentclass[a4paper,twocolumn,10pt,accepted=2025-10-08]{quantumarticle}
\pdfoutput=1

\usepackage[numbers]{natbib}

\usepackage[dvipsnames]{xcolor}

\usepackage[utf8]{inputenc}
\usepackage[english]{babel}
\usepackage[T1]{fontenc}
\usepackage{amsthm}
\usepackage{thmtools} 
\usepackage{mathtools}
\usepackage{physics}
\usepackage{bm} 
\usepackage{bbm} 
\usepackage{braket}
\usepackage{dsfont}
\usepackage{amsmath}
\usepackage{amssymb}
\usepackage{float}

        \newcommand*{\idty}{\mathds{1}}
        \newcommand*{\N}{\mathbb{N}}

        \newcommand*{\R}{\mathbb{R}}
        \newcommand*{\C}{\mathbb{C}}
        \newcommand*{\hil}{\mathcal{H}}
        \newcommand*{\defdot}{\,\cdot\,} 
        \newcommand*{\defcolon}{\,:\,} 
        
        \newcommand*{\bits}{\mathbb{B}^{N}} 
        \newcommand*{\states}{\mathfrak{S}(\hil)} 
        \newcommand*{\purestates}{\mathcal{P}(\hil)}
        \newcommand*{\ub}{U_{B}} 
        \newcommand*{\uc}{U_{C}} 
        \newcommand*{\Ham}{\Delta} 
        \newcommand*{\smv}{\mu} 
        \newcommand*{\lie}{\mathfrak{g}} 
        \newcommand*{\su}{\mathfrak{su}} 
        \newcommand*{\SU}{\mathrm{SU}} 
        \newcommand*{\tb}{\mathcal{T}_{B}} 

\usepackage[colorlinks]{hyperref}
\hypersetup{
    colorlinks = true,
    citecolor=PineGreen,
    linkcolor=MidnightBlue,
    urlcolor=PineGreen
}

\declaretheorem[style=plain]{theorem}

\declaretheorem[style=plain,sibling=theorem]{corollary}

\declaretheorem[style=plain,sibling=theorem]{observation}
\declaretheorem[style=plain,sibling=theorem]{proposition}

\newcommand*{\itheader}[1]{
    \begin{center}
        \textbf{\textsf{#1}}\\
    \end{center}
}

\usepackage{etoolbox}
\apptocmd{\sloppy}{\hbadness 10000\relax}{}{}

\usepackage{silence}
\begin{document}

\title{Deep-Circuit QAOA}

\author{Gereon Koßmann}
\email{kossmann@physik.rwth-aachen.de}
\affiliation{Institute for Quantum Information, RWTH Aachen University, Aachen, Germany}
\author{Lennart Binkowski}
\email{lennart.binkowski@itp.uni-hannover.de}
\affiliation{Institute for Theoretical Physics, Leibniz University Hannover}
\author{Lauritz van Luijk}
\affiliation{Institute for Theoretical Physics, Leibniz University Hannover}
\author{Timo Ziegler}
\email{timo.ziegler@itp.uni-hannover.de}
\affiliation{Institute for Theoretical Physics, Leibniz University Hannover}
\affiliation{Volkswagen AG, Berliner Ring 2, 38440 Wolfsburg}
\author{Ren\'{e} Schwonnek}
\email{rene.schwonnek@itp.uni-hannover.de}
\affiliation{Institute for Theoretical Physics, Leibniz University Hannover}

\begin{abstract}
Despite its popularity, several empirical and theoretical studies suggest that the quantum approximate optimization algorithm (QAOA) has persistent issues in providing a substantial practical advantage.
Numerical results for few qubits and shallow circuits are, at best, ambiguous, and the well-studied barren plateau phenomenon draws a rather sobering picture for deeper circuits.
However, as more and more sophisticated strategies are proposed to circumvent barren plateaus, it stands to reason which issues are actually fundamental and which merely constitute -- admittedly difficult -- engineering tasks.
By shifting the scope from the usually considered parameter landscape to the quantum state space's geometry we can distinguish between problems that are fundamentally difficult to solve, independently of the parameterization, and those for which there could at least exist a favorable parameterization.
Here, we find clear evidence for a 'no free lunch'-behavior of QAOA on a general optimization task with no further structure; individual cases have, however, to be analyzed more carefully.

Based on our analysis, we propose and justify a performance indicator for the deep-circuit QAOA that can be accessed by solely evaluating statistical properties of the classical objective function.
We further discuss the various favorable properties a generic QAOA instance has in the asymptotic regime of infinitely many gates, and elaborate on the immanent drawbacks of finite circuits.
We provide several numerical examples of a deep-circuit QAOA method based on local search strategies and find that - in alignment with our performance indicator - some special function classes, like QUBOs, indeed admit a favorable optimization landscape.
\end{abstract}

\maketitle

\section{\label{section:Introduction}Introduction}

Within recent years, \emph{variational quantum algorithms} (VQAs) \cite{Cerezo2021VariationalQuantumAlgorithms} have become the focus of a significant amount of research.
The intuitive idea of this class of algorithms is to use classical optimization routines in order to variationally combine small building blocks of quantum circuits into bigger ones that may give good solutions to difficult problems. 

Due to its simplicity and versatility, the \emph{quantum approximate optimization algorithm} (QAOA) \cite{Farhi2014AQuantumApproximateOptimizationAlgorithm} is one of the most prominent types of algorithms from the growing family of VQAs.
It was designed to tackle the class of \emph{combinatorial optimization problems} (see \autoref{section:Preliminaries}).
This class includes, e.g., problems like MAXCUT, MAX-$k$-SAT, or the traveling salesperson problem and can, in general, be considered to be of high practical relevance for real world applications.
Algorithms for non-trivial instances of these can already be implemented on systems with only a few dozen of qubits; e.g. \cite[Figure 4]{Harrigan2021QuantumApproximateOptimizationOfNonPlanarGraphProblemsOnAPlanarSuperconductingProcessor} use up to $23$ qubits.
This is why QAOA attracted, beyond a pure academic interest, also a lot of attention by several of the first commercial providers of quantum software.

Despite the high hopes that are connected to these algorithms and their capabilities, a true proof or demonstration of any practical advantage resulting from the application of a VQA like the QAOA to any real world problem is, however, still pending.
Respecting the limitations of existing technology, we neither have large fault tolerant quantum computers nor can we simulate many qubits, a lot of research focus was put on implementations with few qubits and shallow quantum circuits \cite{Willsch2020BenchmarkingTheQuantumApproximateOptimizationAlgorithm} and therefore basically `proofs of concept'.

In this regime the hopes especially put into QAOA seem to face substantial obstacles:
Already in \cite{Farhi2014AQuantumApproximateOptimizationAlgorithm} it was numerically shown, that the MAXCUT problem for $3$-regular graphs is not (reliable) solvable with low-depth circuits; only increased circuit depths yield decent solution quality \cite{Crooks2018PerformanceOfTheQuantumApproximateOptimizationAlgorithmOnTheMaximumCutProblem}.
Low-depth case studies are done in a broad way through the literature leading to the, at least empirical, conclusion that low-depth QAOA circuits will not give reliable results for complicated problem instances \cite{Lykov2022SamplingFrequencyThresholdsForTheQuantumAdvantageOfTheQuantumApproximateOptimizationAlgorithm,Bravyi2020ObstaclesToVariationalQuantumOptimizationFromSymmetryProtection,Hastings2019ClassicalAndQuantumBoundedDepthApproximationAlgorithms}.

In this work we extend the investigation of the potential perspectives and limitations of the QAOA to the regime of deep quantum circuits.
Central questions that guide us on this path are:
What are distinctive features of this regime?
Are there effects and methods that become present for deep circuits that are not possible in a low-depth regime?
\emph{And most importantly}, on what classes of problems could QAOA perform well when circuit depth is not a hard limitation? 

By this work we try to provide at least some clear answers to those questions.
These answers should, whenever possible, admit a certain aspiration of mathematical rigor and will be backed up by clear intuitions of possible mechanisms and numerical case study evidence otherwise.
For a bigger picture we can, however, only make a beginning.

A distinctive feature for the deep-circuit regime concerns the types of classical variational methods that could be employed.
In the deep-circuit regime, we are confronted with a rapidly growing range of classical control parameters.
Here the classical variation routines that are typically employed in low-depth QAOA quickly run out of their efficiency range.
Instead, we will consider local search routines as the characteristic class of variation methods of the deep-circuit regime, since they are naturally suited for optimizations in this situation. 

In the first part of \autoref{section:AsymptoticCircuits}, we give a clarified view on the QAOA optimization landscape on state space that is, to our surprise, unexpectedly seldom employed.
However, it fits well for analyzing local search routines which notably do not admit for a fixed circuit length.
The basic QAOA Hamiltonians are treated as generators of a Lie algebra whereby the optimization landscape is given by an orbit of its Lie group.
The resulting landscape reveals a nice geometry that can be analyzed by basic tools from differential geometry. 
This view is well investigated in the field of optimal control theory from which we borrow methods and results.
Similar methodologies have also been used to further study barren plateaus \cite{Larocca2022DiagnosingBarrenPlateausWithToolsFromQuantumOptimalControl}.

In the second part of \autoref{section:AsymptoticCircuits}, we start our investigations by studying the limit of asymptotic circuits.
Here we find that a generic QAOA instance has many favorable properties: for example, a unique local minimum that gives us the optimal solution of the underlying classical optimization problem.
This vaguely means that a local search that could exploit arbitrary circuit depths and employ second order gradient methods will succeed in solving almost any problem.
This result can be seen in line with findings that were, e.g., reported in \cite{Arenz2021ProgressTowardFavorableLandscapesInQuantumCombinatorialOptimization} showing that variational quantum algorithms with an exponential amount of control parameters can avoid local traps.

These regimes are, however, far from any practical use.
In \autoref{section:DeepCircuits}, we turn our attention to deep, but not asymptotically deep, circuits. Here many of the nice asymptotic properties vanish.
Saddle points turn into effective local minima, and we get a landscape with a continuum of local attractors and potentially exponentially many local traps.
However, a characterization of local traps reveals that statistical distribution properties of traps, like amount, sizes, and depths, only depend on the classical objective function, and can be used to predict the performance of the deep-circuit QAOA in this generically unfavorable setting.
Our method gives strongly problem-dependent quantities that serve as an evaluation basis for the success of the QAOA, and we collect them in a single performance indicator.
Especially the lack of success in many QAOA instances could be explained from this perspective.
As an example for a 'no free lunch'-behavior we come up with, in a very simple way, randomly generated target functions that impose an optimization landscape with unfavorably distributed traps.
In contrast, we see that certain special problem classes, like QUBO, have the tendency to admit a favorable landscape.

In \autoref{section:NumericalResults} we look at our results from a practical perspective by numerically simulating deep-circuit QAOA (up to 1000 layers of non-decomposed QAOA-gates) based on a simple downhill simplex method.
Results indicate the QAOA performs well only on some classes of optimization problems.
Most notably, the numerical results are in line with the prior introduced performance indicator.
Furthermore, we numerically justify the basic intuitions regarding traps and their influence on local search strategies.
An exemplary step size analysis also shows that, unexpectedly, the same problem instance can lead to very different results when approached with different auxiliary parameters.

\section{\label{section:Preliminaries}Preliminaries}

For the reader's convenience we will start with a brief review of the VQA approach, elaborate on the specific form of the QAOA, and outline our view on optimization landscapes, which we will use throughout this work.
Further, we briefly elaborate on the barren plateau phenomenon and its many variants.

\subsection{\label{subsection:VariationalQuantumAlgorithms}Variational Quantum Algorithms}

Consider a generic unconstrained combinatorial minimization problem:
\begin{align}\label{equation:CombinatorialMinimizationProblem}
    \min_{z \in \bits} f(z),
\end{align}
where $\bits \coloneqq \{0, 1\}^{N}$ denotes the set of bit strings of length $N$.
We adhere to the following standard encoding procedure, although more problem-specific and qubit-efficient approaches may be available \cite{Tan2021QubitEfficientEncodingSchemesForBinaryOptimisationProblems,Leonidas2023QubitEfficientQuantumAlgorithmsForTheVehicleRoutingProblemOnNISQProcessors}:
identify each bit string $z$ with a computational basis state $\ket{z}$ of the $N$-qubit space $\hil \coloneqq \C^{2^{N}}$ and translate the objective function $f$ into an objective Hamiltonian, diagonal in the computational basis
\begin{align}\label{equation:Encoding}
    f \mapsto H \coloneqq \sum_{z \in \bits} f(z) \ketbra{z}{z}.
\end{align}

Despite mainly working with pure states, we will advocate to describe the state of a quantum system within the formalism of density matrices, i.e.\ positive semi-definite matrices $\rho \succeq 0$ of unit trace.
We denote the state space by $\states$.
An immediate consequence of using this natural formalism is that the expectation value of an observable $H$ 
\begin{align}\label{equation:Functional}
    F(\rho) \coloneqq \tr(\rho H)
\end{align}
defines a linear functional $F : \states \rightarrow \R$.
Note that by the Rayleigh-Ritz principle, the original minimization task \eqref{equation:CombinatorialMinimizationProblem} is equivalent to finding a minimum of $F$.
For our study of derivatives, critical points, and minima of $F$, this linearity will turn out as very beneficial.

In order to approximately minimize $F$, a general VQA now utilizes parameterized trial states obtained by applying a parameterized quantum circuit (PQC) to an initial state.
For the QAOA, the initial state is given by the pure state $\ketbra{+}{+}$, where 
\begin{align}\label{equation:PlusState}
    \ket{+} = \bigotimes_{n = 1}^{N} \frac{1}{\sqrt{2}} (\ket{0} + \ket{1}) = \frac{1}{\sqrt{2^{N}}} \sum_{z \in \bits} \ket{z}.
\end{align}
It is the non-degenerate ground state of the Hamiltonian\footnote{The original QAOA is formulated for maximization tasks; here $B$ is defined with the opposite sign.} 
\begin{align}\label{equation:B}
    B = - \sum_{n = 1}^{N} \sigma_{x}^{(n)}.
\end{align}
The unitary evolution generated from $B$ is commonly referred to as 'mixing'.

For the QAOA, we have two basic families of gates
\begin{subequations}
\begin{eqnarray}
    \ub(\beta) &=& e^{- i \beta B}\quad\text{and} \label{equation:MixerGate}\\
    \uc(\gamma) &=& e^{- i \gamma C}, \label{equation:PhaseSeparatorGate} 
\end{eqnarray}
\end{subequations}
where we used the convention of taking $C=H-\tr(H)\idty$ as a traceless generator for our second unitary, which is usually coined the 'phase separator'.
Note that taking all generators traceless does not change the evolution on the level of density matrices.

A full QAOA circuit is then composed of these basic families. 
The corresponding parameters are typically labeled by $\vec{\beta} = (\beta_{1}, \ldots, \beta_{p}) \in [0, \pi)^{p}$ and $\vec{\gamma} = (\gamma_{1}, \ldots, \gamma_{p}) \in \R^{p}$ with $p \in \N$, where the quantity $p$ specifies the circuit depth.
A fully parameterized QAOA circuit is therefore given by
\begin{align}\label{equation:QAOAUnitary}
     V(\vec{\beta}, \vec{\gamma}) \coloneqq \prod_{q = 1}^{p} \ub(\beta_{q}) \uc(\gamma_{q})
\end{align}
and gives rise to the parameter function
\begin{align}\label{equation:ParameterFunction}
    \tilde{F}(\vec{\beta}, \vec{\gamma}) \coloneqq F(V(\vec{\beta}, \vec{\gamma}) \ketbra{+}{+} V(\vec{\beta}, \vec{\gamma})^{*} H).
\end{align}

\subsection{\label{subsection:TheDeepCircuitRegime}The Deep-Circuit Regime}

In the majority of the existing literature circuits with a depth $p \approx 10$ or less are taken into account.
One practical reason for this restriction is the still too high gate noise in existing quantum computers.

In this work we  will refer to deep circuits as those that substantially exceed the scale of $p \approx 10$, potentially by orders of magnitudes. For example, the circuit depths used in our numerical simulations ranged from $100$ to $3 \cdot 10^{4}$. From a technological perspective, this regime is not yet reliably realizable on most existing devices, but clearly at the edge of what can be expected from the technological developments in the near and midterm future. The key achievements we are here hoping for are improved gate fidelities that could, e.g.\ be enabled by improved error mitigation techniques \cite{Endo2018PracticalQuantumErrorMitigationForNearFutureApplications} or even a full implementation of (few) error corrected qubits \cite{Egan2021FaultTolerantControlOfAnErrorCorrectedQubit}.

A central ingredient in any VQA are the classical routines involved in optimizing the circuit parameters.
Here the types of classical optimization algorithms that can be used for shallow or deep circuits might differ substantially.
In shallow circuits all variational parameters can be actively optimized at the same time.
Typical optimizers in use are for example COBYLA, Gradient Descent, BFGS, and Nelder-Mead.
Those may however not perform well for deep circuits.
An optimization of a largely growing amount of parameters might quickly become unpractical, either by an increase of computational hardness or by the commonly observed barren plateaus (BP) phenomenon \cite{Larocca2025BarrenPlateausInVariationalQuantumComputing}.

For deep circuits, we will therefore focus on optimization routines that follow local search strategies, i.e. those which successively only vary a few parameters at the same time within a small range of variation.
We tend to mark the necessity of these routines as another distinguishing feature of the deep-circuit regime.
As a naive prototype, we use a very simple downhill simplex method in our numerical studies in \autoref{section:NumericalResults}, being fully aware that a huge variety of very elaborated and versatile local search routines exist.
For layer-wise optimization, BPs have also been reported in some particular instances \cite{Campos2021AbruptTransitionsInVariationalQuantumCircuitTraining}.
These instances, however, substantially differ from our setting, and we can attribute the phenomenon of vanishing gradients to `local traps' instead.

\subsection{\label{subsection:OptimizationLandscapes}Optimization Landscapes}

In the context of the QAOA, the term optimization landscape commonly refers to the 'landscape' of values that the function $\tilde{F}$ takes on a parameter space, which is a subset of $\R^{2 p}$ (or on a two dimensional subspace whenever a graphical illustration is provided).
However, for local search strategies, the final length $p$ of a circuit, and therefore the parameter space, is typically not fixed.
Hence, the above notion of 'optimization landscape' does not really apply here.
This motivates us to think of optimization landscapes in a different way: 

We simply regard the functional $F$ as a function that defines a 'landscape' on the state space $\states$, or more precisely, on the subset of the states that are potentially accessible by a sequence of QAOA gates.
To our big surprise, this perspective is rather rare to find within the existing literature on VQAs.
Its indeed very nice geometry will therefore be fully clarified within the next section.

In order to properly distinguish notions we will from now on refer to landscape in $\R^{2 p}$ as the \emph{parameter landscape} of $\tilde{F}$ and to the latter one as the \emph{state space landscape} of $F$.
Using the state space landscape to determine properties of an QAOA algorithm has at least three immediate advantages: 
\begin{itemize}
    \item[(i)] Circuits with varying length can be properly expressed and compared within the same picture. 
    \item[(ii)] Local overparameterization is avoided. Different sets of parameters could refer to almost the same point in state space in no obvious way. This can, for example, lead to a situation in which there are many different local minima in the parameter landscape that however only correspond to one and the same minimum in the state space landscape. 
    \item[(iii)] The behavior of different classes of VQAs can be compared within the same picture.  
\end{itemize}
In this work, especially the points (i) and (ii) will be essential for analyzing the performance perspectives of deep QAOA circuits by characterizing the distribution properties of local minima and critical points.

\subsection{\label{subsection:BarrenPlateaus}Barren Plateaus}

In contrast to the state space landscape, the prominently featured parameter landscape is more accessible from a pure practical point of view.
The parameter values and their updates are provided by the classical optimizer and are thus easy to record throughout the execution of a VQA while the state space landscape is not directly observable.
Despite being a difficult to handle mathematical object ($\tilde{F}$ is generally highly non-linear in the parameters), there have been many impressive breakthroughs determining some of its properties; first and foremost BPs, which were first discovered in \cite{McClean2018BarrenPlateausInQuantumNeuralNetworkTrainingLandscapes}.
In essence, it has been shown that for PQCs that cover the state space sufficiently uniformly, the expectation values of the gradients of $\tilde{F}$ in parameter direction are zero while the variance decreases exponentially with the number of qubits: $\text{Var}(\delta_{k} \tilde{F}) \in \mathcal{O}(b^{n})$, for some $b > 1$.
Therefore, starting from random initial parameters, exponentially many measurements have to be taken in order to record a slope:
the parameter landscape becomes practically flat almost everywhere which does not only effect gradient-based methods, but also gradient-free optimization routines \cite{Arrasmith2021EffectOfBarrenPlateausOnGradientFreeOptimization}.
A sophisticated parameter initialization method could, however, still circumvent this issue.

The initial discovery of BPs spawned an extensive and revealing search for other sources of BPs.
Soon, BPs where reported to also emerge from quantum device noise for PQCs growing linearly with the number of qubits \cite{Wang2021NoiseInducedBarrenPlateausInVariationalQuantumAlgorithms}, disconnecting the effect from random parameter initialization.
While this issue could be tackled with more developed quantum error correction methods, BPs where also shown to emerge from entanglement \cite{OrtizMaerrero2021EntanglementInducedBarrenPlateaus}, a property typically considered a feature rather than a drawback.
Another finding establishes the existence of BPs even for shallow circuits, provided that the cost function $\tilde{F}$ corresponds to the expectation value of global observables \cite{Cerezo2021CostFunctionDependentBarrenPlateausInShallowParametrizedQuantumCircuits}.
Additionally, BPs were found to contain many local minima \cite{Nemkov2025BarrenPlateausSwampedWithTraps}, further complicating maneuvering through the parameter landscape.

All these effects can be interpreted as a manifestation of the \emph{curse of dimensionality}:
The exponentially large dimension of the underlying Hilbert space makes it generally exponentially hard to extract any directional information.
However, every variant of BPs comes with its own specific assumptions on either PQC properties, circuit depth, or the implementation of the objective function.
In the following we aim at dropping more and more structural requirements and focus on the underlying optimization problem itself.

\section{\label{section:AsymptoticCircuits}Asymptotic Circuits}

As starting point of our investigation we will discuss the perspectives of the QAOA in an asymptotic regime.
This is, we consider infinite sequences of gates with potentially infinitesimally small parameters.
By including infinitesimal elements, the analysis of the asymptotic QAOA becomes drastically more structured since we now can make direct use of tools from differential geometry, i.e. Lie groups and their algebras. 

In the first part of this section we will clarify this geometry. Here the key points are: 
\begin{itemize}
    \item Accessible states form a differentiable manifold $\Omega$
    \item The target functional corresponds to a differentiable scalar field on $\Omega$
    \item The two families of QAOA gates correspond to two vector fields on $\Omega$
\end{itemize}

We will then turn our attention to the classification of minima and critical points of $F$.
Here our key findings for generic instances are: 
\begin{itemize}
    \item We have a one-to-one correspondence between the possible solution space $\mathbb{B}^N$ of the classical optimization problems and critical points of $F$.
    \item Within these critical points there is only one local minimum.
    This minimum is the solution of the underlying optimization problem. 
\end{itemize}

Thematically, our investigations and results belong to the field of dynamical Lie algebras (see e.g. \cite{Ge2022TheOptimizationLandscapeOfHybridQuantumClassicalAlgorithmsFromQuantumControlToNISQApplications} for a comprehensive review and \cite{Allcock2024OnTheDynamicalLieAlgebrasOfQuantumApproximateOptimizationAlgorithms} for concrete applications to QAOA).
However, we consider our analysis as more tangible than the general case discussed there.

\subsection{\label{subsection:TheGeometryOfAccessibleStates}The Geometry of Accessible States}

The very first question that arises when considering asymptotic gate sequences is:
Which states can be reached by the QAOA when starting from an initial state $\psi_{0}$?
This set of accessible states, from now on denoted by $\Omega$, will be the ground for the geometrical picture we want to outline in this section.  

Let $H$ be the target Hamiltonian that encodes an optimization problem and let $U_B(\beta)$ and $U_C(\gamma)$ with $C=H-\tr(H)\idty$ be our basic families of QAOA unitaries as described in the previous section.
The set $\Omega$ will be obtained from considering the set $G(B,C)$ of all circuits that could be asymptotically generated by the QAOA.  
In the following, we will drop the explicit dependence on $B$ and $C$ and write $G=G(B,C)$ whenever it is clear from the context.

\itheader{Circuits}
As described in the previous section, circuits of a fixed depth $p$ are parameterized by finite sequences of angles $\vec\gamma$ and $\vec\beta$.
For any $p$, these can be captured by the set
\begin{align}
   G_{p}= \left\{ 
    \prod_{q = 1}^{p} U_{C}(\gamma_{q}) U_{B}(\beta_{q}) \in\SU\big(2^{N}\big) : (\vec\gamma,\vec\beta) \in \mathbb{R}^{2p} \right\}.\label{equation:FixedDepthGroup}
\end{align} 
There are now some technicalities to respect when considering circuits of infinite depth.
A naive limit $p \rightarrow \infty$ of infinite angles sequences in \eqref{equation:FixedDepthGroup} will lead to infinite products of unitaries, which is not necessarily a well-defined object.

However, on the level of sets we observe that the $G_{p}$ form a monotone sequence; i.e.,\ we have $G_{p} \subseteq G_{p + 1}$.
Here a well-defined set-theoretic limit exists.
From this limit we obtain $G$ by taking the topological closure in $\mathcal L (\mathcal H)$, i.e., 
\begin{align}\label{equation:AsymptoticGroup}
   G=
    \overline{\Big(\lim_{p \to \infty} G_{p}\Big) } \subset \SU(2^N).
\end{align} 
Intuitively, $G$ contains all circuits of finite depth as well as all unitary transformations that can be approximated by circuits of finite depth up to arbitrary precision.
Here proximity is measured in the operator norm, meaning that two unitaries will be close to each other whenever their images are close to each other for any input state. 

\itheader{States}
Applying these circuits to a fixed initial state $\Psi_{0}$, i.e.\ taking the $G$-orbit around $\Psi_{0}$, will then give us the set $\Omega$ of accessible states
\begin{align}\label{equation:AccessibleStates}
    \Omega:= \left\{
    U \Psi_{0} U^{*}
    \defcolon
    U \in G
    \right\}\subseteq \states .
\end{align}
Corresponding to our intuition on $G$, those are all states that can be generated from $\Psi_{0}$ by a finite circuit such as all states that are approximately close to them. 
The central observation for the following analyses is that the set $G$ naturally carries the structure of a Lie group, which will also carry over to a corresponding structure on $\Omega$.

On one hand, a Lie group has the structure of a group, which in this case simply reflects that QAOA circuits have some of the basic properties a complete set of quantum circuits should have.
The group action, here given by multiplying the corresponding unitaries, reflects that concatenating two possible circuits gives again a valid circuit.
The existence of a neutral element, the identity operator, is obtained by setting all angles $\vec{\gamma}$ and $\vec{\beta}$ to zero.
This corresponds to an empty circuit, i.e.\ doing nothing.
The existence of an inverse\footnote{This is clear for finite circuits, for infinite circuits the proper limits have to be taken into account.}, here provided by taking the adjoint of a unitary, reflects the often-advertised property that quantum circuits are reversible and that the set of QAOA circuits is complete with respect to this property.  

On the other hand, the Lie group $G$ has the structure of a manifold.
This structure will be crucial for rigorously talking about optimization landscapes in what follows.
Even though the manifold structure of a Lie group is given in a quite abstract manner in the first place\footnote{On an abstract level, we can consider $G$ as a subgroup of $\mathrm{GL}(n,\C)$.
Here, Cartan's theorem on closed subgroups \cite{Lee2003IntroductionToSmoothManifolds} states that the embedding of $G$ into the smooth structure of $\mathrm{GL}(n,\C)$ will give us a consistent smooth structure on $G$ itself.},
it will carry over to a concrete structure on $\Omega$, giving us the playground for defining optimization landscapes and analyzing the behavior of optimization routines.  

For $\omega \in \Omega$ and $U \in G$, the map 
\begin{align}
  \pi_{\omega}(U) := U \omega U^{*} \label{equation:ActionOnStates}
\end{align}
describes an action of $G$ on $\Omega$ and defines the Lie group structure on $\Omega$.
Firstly, 
\begin{align}
     \pi_{\pi_{\omega}(V)}(U)= U V \omega V^{*} U^{*} = \pi_{\omega}(U V)\label{equation:GroupHomomorphism}
\end{align}
gives us a group homomorphism.
Secondly, as outlined later on, this map will also carry over the manifold structure from $G$ to $\Omega$.
As a manifold, $\Omega$ is isomorphic to the Lie group quotient $G / N$ via the map $\pi_{\Psi_{0}}$ where $N = \{U \in G \defcolon U \Psi_{0} U^{*} = \Psi_{0}\}$.
A direct consequence is that $\Omega$ is a closed manifold, i.e.\ is compact and has no boundary.

\begin{figure}[!ht]
    \includegraphics[width=0.99\linewidth]{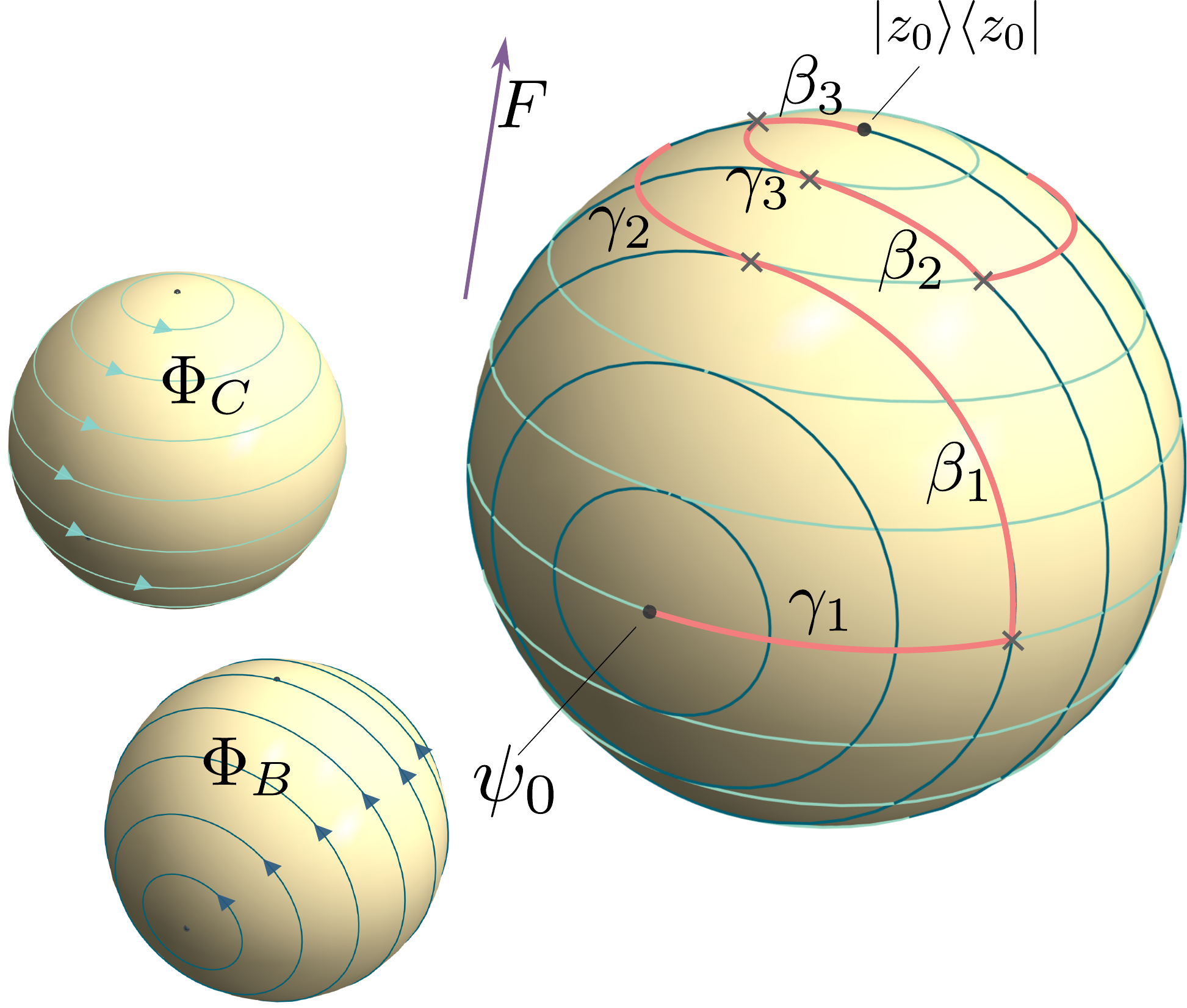}
    \caption{The geometry of the QAOA visualized on a qubit:
    Here, $\Omega$ is given by the surface of the Bloch sphere (all pure qubit states).
    Applying the basic gates $U_C$ or $U_B$ corresponds to movements along vector fields $\Phi_C$ or $\Phi_B$ that are oriented along lines of constant latitudes (left).
    Performing a QAOA sequence (coral colored line) corresponds to alternatingly move along these vector fields.}
    \label{figure:BlochSphere}
\end{figure}

\itheader{The Objective}
Our target functional $F(\omega) = \langle H \rangle_{\omega}$ can now be seen as a scalar field $F : \Omega \rightarrow \R$. 
Our initial minimization task is hence equivalent to finding the minimum of this field on $\Omega$. 
A direct consequence that can be drawn from the fact $\Omega$ is a closed manifold is that local and global minima only occur in the interior.
It is easy to see that $F$ is indeed also differentiable, which directly implies that local minima can be characterized considering its derivatives. 

Derivatives on a manifold are given in terms of tangent spaces. 
The tangent spaces on $G$ are obtained from its Lie algebra.
We will denote the Lie algebra that corresponds to $G$ by $\lie(B,C)$.
It always contains the Lie algebra generated by $i B$ and $i C$:
\begin{align}\label{equation:LieAlgebra}
    \lie(B,C) &\supseteq \operatorname{Lie}\left(\{i B, i C \}, [\defdot, \defdot]\right)\nonumber\\
    &=\mathrm{span}_{\mathbb{R}}\hspace*{-2pt}\left\{i B, i C, [B, C], i [B, [B, C]], \ldots\right\}.
\end{align}
The question whether the inclusion in \eqref{equation:LieAlgebra} is actually an equality is subtle, but not too important for our analysis.
We will drop the explicit dependence on $B$ and $C$ in the notation and write $\lie = \lie(B,C)$ whenever it is clear from the context.
Our choice of taking $B$ and $C$ to be trace-less implies that $\lie \subseteq \su(\hil)$.
The map $\pi_{\omega}$ now transfers the tangent spaces of $G$ to tangent spaces on $\Omega$.
More precisely, for a point $\omega \in \Omega$, the map $\pi_{\omega}$ gives us a push forward from $\lie$ to the tangent space $T_{\omega} \Omega$ at $\omega$.

Resulting from this we can introduce the directional derivatives of a scalar field w.r.t.\ $A \in \lie$ at a point $\omega \in \Omega$ by
\begin{align}
    \nabla_{A}^{(n)} F(\omega) &:= \partial_{t}^{(n)} F\big(\pi_{\omega}(e^{t A})\big)\vert_{t = 0}\nonumber\\
    &\phantom{:}= \partial_{t}^{(n)} F\big(e^{t A}\, \omega\, e^{- t A}\big)\vert_{t = 0}.\label{equation:Derivative}
\end{align}
These derivatives will be used to characterize critical points, minima, and maxima in the subsequent sections.
At each point $\omega \in \Omega$, there will be several $A \in \lie$ that do not correspond to actual directions on $\Omega$ due to its quotient structure. 
Namely, for $A\in \lie$ with
\begin{align}
    [\omega,A] = 0,\label{equation:VanishingCommutator}
\end{align}
we have that $\pi_{\omega}(e^{t A}) \equiv \omega$ for all $t \in \R$.
Thus, the directional derivatives w.r.t.\ $A$ would vanish for every considered scalar field and is therefore meaningless.

\itheader{QAOA Sequences}
To complete this section we will have a look at the geometric role that the families of basic QAOA gates $U_{B}(\beta)$ and $U_{C}(\gamma)$ play. 
From the manifold perspective those correspond to vector fields $\Phi_{B}$ and $\Phi_{C}$ on $\Omega$. 
Explicitly, we can think of vectors
\begin{align}
    &i [\omega, B] \text{ and } i [\omega, C]\label{equation:Directions}
\end{align}
that are assigned to each point $\omega \in \Omega$.
The structure of these vectors can be grasped in a simple manner by considering that a point $\omega$ is transported by a one parameter family $U_{t} = e^{t A}$ via the map \eqref{equation:ActionOnStates} along a continuous curve  $\omega_{t} = \pi_{\omega}(e^{t A})$ on $\Omega$.
The directional derivative along this curve is then simply given by 
\begin{align}
    \partial_{t} \omega_{t}\vert_{t = 0} = \partial_{t} e^{t A}\, \omega\, e^{-t A}\vert_{t = 0} = - [\omega, A].\label{equation:EvaluatedDerivative}
\end{align}
Performing a QAOA instance of depth $p$ with some parameters $\beta_{1}, \beta_{2}, \dots, \beta_{p}$ and $\gamma_{1}, \gamma_{2}, \dots, \gamma_{p}$ hence can be understood in a nice geometrical picture: start from $\Psi_{0}$ and follow (the flow of) the vector field $\Phi_{C}$ for a distance given by parameter $\gamma_{1}$, stop and change direction to follow the (flow of the) vector field $\Phi_{B}$ for a distance given by $\beta_{1}$, and so on.
Here the circuit depth $p$ determines the number of switching from one flow to the other.
Situations in which shallow circuits suffice to reach a global minimum can therefore be associated to situations in which the vector fields have a simple structure (see, e.g., \autoref{figure:BlochSphere}).
In converse, we expect that the, potentially exponential, hardness of a particular classical problem translates into a corresponding complexity of the vector field structure.
Here the dimension of the corresponding Lie algebra would be a natural number for quantifying this. 

\subsection{\label{subsection:TheLieAlgebraOfTheEncodedProblem}The Lie Algebra of the Encoded Problem}

For the characterization of optimization landscapes, we will first focus on the underlying manifold.
Here, the central question is: which Lie group corresponds to a classical optimization problem encoded in a given Hamiltonian $H$? 
As it turns out, all properties of a Lie group which are relevant for our investigations are already encoded in its Lie algebra so that we focus on these objects instead.
Here, the above question translates to determining properties of the algebra $\lie$, which is generated by the given operators $B$ and $C$. 

The precise form of $\lie$ might, however, be highly problem-specific and will reflect the hardness of the underlying classical optimization problem becomes, in general, extremely hard to compute for large systems.
At this point, we will leave the detailed investigation of this connection for future work and instead provide statements on $\lie$ for a generic case. 

Notably, it was prominently shown \cite{Lloyd1995AlmostAnyQuantumLogicGateIsUniversal,Deutsch1995UniversalityInQuantumComputation} that almost every set of quantum logic gates with an action on more than one qubit can be used to generate any quantum circuit.  
A corresponding result for Lie algebras states that the algebra generated by randomly drawn traceless matrices almost certainly turns out to be the full algebra $\mathfrak{su}(2^{N})$.
This clearly suggests to ask for a similar behavior in the special case of the QAOA circuits. 

Formally, we will say that the generators $B$ and $C$ of QAOA-gates are \emph{universal} if $G(B,C) = \SU(\hil)$. This follows automatically if the Lie algebra generated by $i B$ and $i C$ is $\su(\hil)$.
Clearly, this question has been investigated before \cite{Morales2020OnTheUniversalityOfTheQuantumApproximateOptimizationAlgorithm} with the result that universality of QAOA circuits was proven for a wide range of objective Hamiltonians $H$ and their corresponding generator $C$.
The following theorem will contribute to this by a simple sufficient criterion for universality that can easily be checked by merely considering the possible values of the classical objective function $f$. 
In the field of optimal control theory the notion of `controllability' is the counterpart to the universality of a gate set in our case. 
Based on a convenient controllability criterion from \cite{Altafini2002ControllabilityOfQuantumMechanicalSystemsByRootSpaceDecompositionOfSuN} we get

\begin{theorem}\label{theorem:UniversalGates}
    $B$ and $C$ together form universal generators of QAOA-gates if the underlying classical optimization problem given by a target function $f$ fulfills the conditions 
    \begin{enumerate}
        \item[(a)] non-degenerate values: $$f(z) = f(z') \implies z=z'$$
        \item[(b)] non-degenerate resonance: $$f(z) - f(z') = f(t) -f(t') \implies (z,z') = (t,t') $$
        \text{ if } $z\neq z' \text{ and } t\neq t'$,
    \end{enumerate} 
    
    In particular, the set of optimization problems $f$ for which $\lie=\su(2^{N})$ is open and dense (and hence the complement is a null set).
\end{theorem}

\begin{proof}
Given Hermitian operators $B$ and $C$ on a (finite dimensional) Hilbert space $\hil$, under what conditions is the generated Lie algebra $\mathrm{Lie}(i B,i C)$ equal to $\su(\hil)$? This question is particularly important in the field of {\it control theory} and has, for example, been looked at in  \cite{Altafini2002ControllabilityOfQuantumMechanicalSystemsByRootSpaceDecompositionOfSuN,Zeier2011SymmetryPrinciplesInQuantumSystemsTheory,Zeier2015OnSquaresOfRepresentationOnCompactLieAlgebras,Zimboras2015SymmetryCriteriaForQuantumSimulabilityOfEffectiveInteractions}. 
We will use the sufficient condition in \cite[Theorem 2]{Altafini2002ControllabilityOfQuantumMechanicalSystemsByRootSpaceDecompositionOfSuN} stating that if $C$ is strongly regular (which is equivalent to our assumptions (a) and (b)) and if the graph $\mathcal{G}$ of $B$ is connected (see below), then $\lie(B,C) =\mathfrak{su}(\hil)$, or equivalently $\lie(B,C) =\su(\hil)$.

The graph $\mathcal{G}$ of an operator $B$ is defined with respect to a basis $\{\ket{z} \in S\}$ as follows:
The vertices of $\mathcal{G}$ are the different basis labels $s\in S$ and there is an oriented edge joining $z$ and $z'$ if and only if $\braket{z | B | z'} \neq 0$.
That is, the matrix representation of $B$ in the same basis $\{\ket{z} \in S\}$ is, up to normalization of its entries, the adjacency matrix of the graph $\mathcal{G}$.
However, since $B$, as a matrix in the computational basis, is irreducible \cite{Farhi2014AQuantumApproximateOptimizationAlgorithm}, we readily obtain that its graph is connected.
\end{proof}

For the case that $\lie = \su(2^N)$, it is indeed always possible to solve the problem in finitely many steps, i.e., there exist $\beta_j,\gamma_j$, $j=1,\ldots,k$ such that $U_B(\beta_1)U_C(\gamma_1)U_B(\beta_2)\cdots U_C(\gamma_k)$ maps the initial state to a ground state of $H$ (see \cite[Thm. 3.4]{Agrachev2004ControlTheoryFromTheGeometricViewpoint}).

Thus, a randomly chosen optimization target $f$ will fulfill the criteria from \autoref{theorem:UniversalGates} almost certainly.
We can therefore draw the conclusion that universality is indeed a generic property of the QAOA.
Generic instances of Knapsack, where the values and weights are arbitrary real numbers, or Traveling Salesman, where the distances between cities are arbitrary real numbers, (both encoded with soft constraints) will fall in this category.
Furthermore, generic QUBOs are one of the classes for which universality was shown in \cite{Morales2020OnTheUniversalityOfTheQuantumApproximateOptimizationAlgorithm}. 

Problem classes where values of $f$ are typically given by integers, like MAXCUT or MAX-$k$-SAT, do however not necessarily fulfill the conditions of \autoref{theorem:UniversalGates} and could therefore lead to an optimization on smaller sets.
If universality is a desired property, it can however always be restored by adding small perturbation.
A concrete strategy would be to consider weighted MAXCUT or weighted MAX-$k$-SAT, with close to integer weights.
The universality in these cases is ensured by the denseness statement in \autoref{theorem:UniversalGates}.
Active perturbation might not even necessary in practice, as basically any numerical inaccuracies stemming from finite machine-precision will lift degenerate resonances almost surely.
Therefore, from a practical point of view, any instance with non-degenerate values admits universality.

In general, it is however far from obvious whether a lack of universality is an obstacle or an actual feature.
On one hand, we have that the reachable set $\Omega$ becomes $\purestates$, the whole set of pure states when $\lie=\su(2^{N})$ and thus, the largest possible search space.
A smaller algebra would in turn lead to a smaller search space, which may make finding an optimum an easier task.
However, on the other hand, a bigger algebra in principle also gives more paths that could be taken by the QAOA, which may in contrast enhance our chances to find short paths, which is one of the QAOA's advertised features. 
Another feature of universality that we will explore in the following is that a sufficiently big algebra will prohibit the existence of local minima in the optimization landscapes.
Here the intuition would be that a larger algebra provides us with more directions to move towards in order to escape from local extrema.

\subsection{\label{subsection:ClassificationOfCriticalPoints}Classification of Critical Points}

With the basic geometry set up in the previous subsections we can now turn our attention to our intended task of minimizing the target functional $F$.
As mentioned above, it is clear from the underlying manifold structure that the global minimum of $F$ on $\Omega$, the point we want to find, does not occur on any boundary.
Hence, we will now follow the usual procedures for discussing the extreme points of a differentiable function.
We start by considering critical points:

We call $\omega_{0} \in \Omega$ a \emph{critical point} of $F$ provided that $\nabla_{A} F(\omega_{0})$ vanishes for all $A \in \lie$.
It is clear that any minimum has to be within the set of critical points.
Furthermore, the existence of critical points also plays a major role in deep-circuit QAOA based on local optimization strategies.
Here critical points typically appear as local attractors, which may impose major hurdles for a good overall performance. 
In the universal case, the critical points can be identified as the eigenstates of $H$.

\begin{proposition}\label{proposition:EigenstatesAreCriticalPoints}
Let $B$ and $C$ be universal generators of QAOA-gates.
Then the critical points of $F$ are precisely the eigenstates of $H$.
\end{proposition}

\begin{proof}
For an arbitrary $\omega\in \Omega$ and $A \in \lie(B, C)$, it holds that
\begin{eqnarray}
    \nabla_{A} F(\omega) 
    &=& \partial_{t} \tr\left(
    H e^{t A}\, \omega\, e^{- t A}
    \right)\vert_{t = 0}\nonumber\\
    &=& \tr(H A\, \omega) - \tr(H \omega A)  \nonumber\\
    &=& \tr(A\, \omega H) - \tr(A H \omega) \nonumber\\
    &=& \tr(A\, [\omega, H]).\label{equation:CriticalDerivative}
\end{eqnarray}
Now $\omega$ is an eigenstate of $H$ if and only if the commutator $[\omega, H]$ vanishes.
Thus, all eigenstates of $H$ that lie in $\Omega$ are critical points.
However, since we assume $B$ and $C$ to be universal, $\Omega = \purestates$ holds.
Conversely, let $\tr(A [\omega, H]) = \tr(\omega [A, H])$ vanish for all $A \in \lie(B, C)$.
Due to the assumed universality of $B$ and $C$, we can choose $A = [\omega, H]^{*}$ and conclude that $\norm{[\omega, H]} = 0$.
Therefore, $\omega$ is already an eigenstate of $H$.
\end{proof}

If the classical target function $f$ fulfills the conditions of \autoref{theorem:UniversalGates}, all eigenstates of $H$ are non-degenerate and therefore coincide with the computational basis states. Thus, we conclude

\begin{corollary}\label{corollary:ComputationalBasisStatesAreCriticalPoints}
Let $f$ fulfill the conditions of \autoref{theorem:UniversalGates}.
Then the critical points of $F$ are precisely the computational basis states.
\end{corollary}

The last proposition tells us, that the state we are looking for is a critical point of the functional.
Moreover, there are no irregularities in terms of `hidden' minima for some special functional.

\subsection{\label{subsection:UniquenessOfMinimaAndMaxima}Uniqueness of Minima and Maxima}

Next, we want to go one step further and investigate properties of the second derivatives.
These will allow us to distinguish between local minima, maxima, and saddle points.
In fact, we have the following necessary condition for local extrema.
\begin{corollary}\label{corollary:ExtremumCharacterization}
If $\omega_{0} \in \Omega$ is a local minimizer of $F$ then $\omega_{0}$ is a critical point of $F$ and for all $A \in \lie$, it holds that $\nabla_{A}^{(2)} F(\omega_{0}) \geq 0$.
\end{corollary}

We call any critical point of $F$ with indefinite second derivatives  \emph{saddle point}.
In the universal case, the critical points of $F$ are precisely given by the eigenstates of $H$ due to \autoref{proposition:EigenstatesAreCriticalPoints}.
This allows us to classify the local minima even further.
Namely, local minima are already global ones.
\begin{proposition}\label{proposition:LocalMinimumIsGlobal}
    Let $B$ and $C$ be universal generators of QAOA-gates.
    Then each local minimum of $F$ is already its global minimum and corresponds to a ground state of $H$.
\end{proposition}
\begin{proof}
By the Rayleigh-Ritz inequality, all global minima of $F$ correspond to ground states of $H$.
Let $\lambda_{0}$ denote its ground state energy.
Now, let $\ketbra{\psi}{\psi} \in \Omega$ be a local minimizer of $F$.
By \autoref{proposition:EigenstatesAreCriticalPoints} and \autoref{corollary:ExtremumCharacterization}, $\ketbra{\psi}{\psi}$ is an eigenstate of $H$; let $\lambda$ denote the corresponding eigenvalue.
Furthermore, for all $A \in \lie(B, C)$, it holds that
\begin{align}
    0 \leq \nabla_{A}^{(2)} F(\ketbra{\psi}{\psi})
    &= \partial_{t}^{2} \tr\left(e^{t A} \ketbra{\psi}{\psi} e^{- t A} H\right)\vert_{t = 0}\nonumber \\
    &= 2 \tr\big(\ketbra{\psi}{\psi} (A^{2}H - A H A)\big)\nonumber \\
    &= 2 \tr\big(A \ketbra{\psi}{\psi} A^{\phantom{*}} (\lambda \idty - H)\big)\nonumber \\
    &= 2 \tr\big(A \ketbra{\psi}{\psi} A^{*} (H - \lambda \idty)\big)\nonumber \\
    &= 2 \big(\hspace*{-2pt}\braket{A \psi | H | A \psi} - \lambda \braket{A \psi | A \psi}\hspace*{-2pt}\big).\label{equation:SecondDerivative}
\end{align}
If $\ketbra{\psi}{\psi}$ would not be a ground state of $H$, one could choose $A \in \lie(B, C) = \mathfrak{su}(2^{N})$ so that $A \ketbra{\psi}{\psi} A^{*}$ is an (unnormalized) ground state of $H$, implying that $\lambda_{0} \geq \lambda$ which clearly rises a contradiction.
\end{proof}

\begin{figure*}[!htp]
    \centering
    \includegraphics[width=0.99\linewidth]{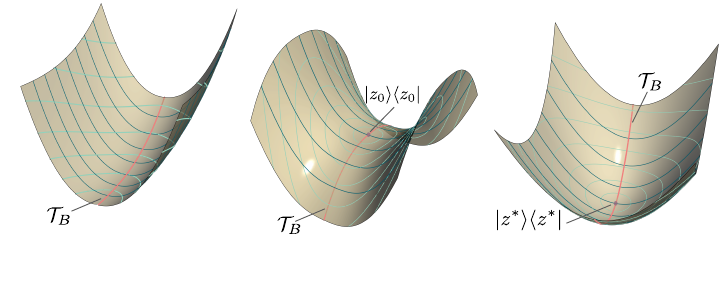}
    \caption{Schematic visualization of different troughs without any computational state, including a local extremum $\ketbra{z_{0}}$, and the optimal state $\ketbra{z^{*}}$, respectively.
    Cyan lines indicate orbits of movements in $i C$ direction. They are also the level sets of the functional. Blue lines indicate orbits of movements in $i B$ direction.
    A component of the \emph{trough}, i.e., a set of states with vanishing first and positive second derivative in $i B$ direction, is marked with red.
    }
    \label{figure:Troughs}
\end{figure*}

If the classical target function $f$ fulfills the conditions of \autoref{theorem:UniversalGates}, the degeneracy of the global minima is lifted and we obtain
\begin{corollary}\label{corollary:IsolatedGlobalMinimum}
    Let $f$ fulfill the conditions of \autoref{theorem:UniversalGates}. Then $F$ admits its only local and global minimum at $\ket{z_{0}} \bra{z_{0}}$, where
    \begin{align*}
        z_{0} = \mathop{\mathrm{arg\ min}}_{z \in \bits} f(z).
    \end{align*}
\end{corollary}

In summary, a generic classical optimization problem, that is, with objective function $f$ fulfilling the conditions of \autoref{theorem:UniversalGates}, induces a very simple optimization geometry on $\Omega$: the critical points of $F$ precisely are the computational basis states and the only local and also global minimum is obtained at the state corresponding to the optimal solution to $f$.
Conversely, at any other point $\omega \in \Omega$, we can always find a direction $A\in \lie$ along which the functional $F$ strictly decreases.
This simple structure is partially lost in the following chapter, where we only allow certain directions.
The loss of information results in the appearance of additional local minima.

Our previous results are to be compared with \cite{Arenz2021ProgressTowardFavorableLandscapesInQuantumCombinatorialOptimization}:
While we show the absence of local minima in the state space landscape, they prove that certain VQA ansätze for MAXCUT result in local minimum-free parameter landscapes.

\section{\label{section:DeepCircuits}Deep Circuits}

We will now discuss the regime of deep circuits.
Here we consider local search routines on circuits with large but not fixed $p$. 
For this, we can keep in mind an intuitive picture of what a local search does: In search of the minimum of $F$, we maneuver through the state space landscape by small steps.
In doing so, we are limited to evaluate $F$ only pointwise and locally.

In contrast to the asymptotic regime from the previous section, we no longer assume that we can move in any direction from $\lie$.
In the explicit case of the QAOA, the remaining directions are $i B$ and $i C$, since moving in any another direction, e.g., $[B, C]$, would in principle require an infinite sequence of infinitesimal gates. 

Additionally, it is  not guaranteed that all states from $\Omega$ can be reached by circuits with a depth $p$ of reasonable order of magnitude.

\subsection{\label{subsection:LocalAttractorsAndTraps}Local Attractors and Traps}

As a direct consequence of restricting search movements to the directions $i B$ and $i C$, it can happen that we loose the ability of escaping from a saddle point. 
This effectively introduces new local minima when transitioning from the asymptotic to the deep-circuit regime.
As we will see, those can, in fact, be uncountably many. 

As before, we can identify local minima by looking at first and second derivatives in the available movement directions. 
Since $F$ stays constant under transformations in $i C$ direction, we only have to consider directional derivatives in the direction $i B$. 
At a state $\phi$ those are given by
\begin{align}
    \nabla_{i B}^{(1)} F \left(\phi\right)= i \tr\left(\phi \; [H, B] \right) \eqqcolon \tr\left(\phi \; F_{1}^{B}\right)\label{equation:f1}
\end{align}
and 
\begin{align}
    \nabla_{i B}^{(2)} F \left( \phi\right) = \tr\left( \phi \;[B,[H, B]] \right) \eqqcolon \tr\left( \phi \; F_{2}^{B} \right).\label{equation:f2}
\end{align}
We substitute the operators $F_{1}^{B}$ and $F_{2}^{B}$, as above, to shorten notations in the following. 
A state will turn into a local minimum for deep circuits if the first derivative vanishes and the second derivative is positive.
This corresponds to the semi-definitely constrained set 
\begin{align}
    \tb \coloneqq \{\phi\in\Omega : \tr\left( \phi \; F_{1}^{B} \right) = 0,\, \tr\left(\phi \; F_{2}^{B}\right) > 0\}\label{equation:Trough}.
\end{align}

Such a set will typically contain a full continuum of points. To get an intuition for this, we can embed $\Omega$ into the Hilbert-Schmidt space. Here, fulfilling the conditions for membership in \eqref{equation:Trough} corresponds to the intersection of the state space with a high dimensional half space and a high dimensional subspace.
Note that, in the generic case, $\mathcal{T}_B$ will be disjoint. In the following, we will refer to its conjoint components as \emph{troughs}.
\begin{observation}\label{observation:Trough}
    Troughs are attracting regions for QAOA with local search routines.
    In, particular, once it comes close to a trough, local search cannot leave the vicinity of the trough unless global movements are performed.
\end{observation}

We will call these global movements \emph{jumps}.
On a purely empirical level, we observe this behavior throughout our numerical studies.
More details are given in \autoref{section:NumericalResults}. 
On a theoretical level, a fully rigorous proof for this is unfortunately difficult to provide, since the full body of local search routines is hard to capture in a mathematical statement. 
We can, however, give some clear theoretical intuitions: 
By construction, we have that, when starting from an arbitrary state $\phi$, the minimization of $F$ along a trajectory $\pi_{\phi}(e^{i \beta B})$, this is, only moving along direction $i B$, will eventually end up in the vicinity of $\tb$.
For a simple local search routine, this is indeed likely to happen, since following this direction will guarantee a monotonous descend.
However, for line search methods, such as \emph{gradient descent}, the delicate challenge of finding a step length that prevents slow convergence on the one and overshooting on the other hand \cite{Nocedal1999NumericalOptimization}, complicates reaching $\tb$ exactly.
Grid search methods, like \emph{hill-climbing}, may even get stuck at states with non-vanishing first derivative because the predetermined step size does not allow for any beneficial steps.

Performing $i C$ in between the $i B$ movements may diminish this behavior, forcing the local search onto a more complicated path.
But if, in the end, this routine also minimizes the gradient in the parameter landscape, it will ultimately converge to $\tb$ as well.
\begin{observation}\label{observation:ZigZag}
    Alternating between $i B$ and $i C$ movements causes the local search to zigzag around the true descent direction of the functional.
\end{observation}

This behavior is typical for local search routines \cite{Russell2009ArtificialIntelligenceAModernApproach} and we also observe it throughout our numerical studies.

\begin{figure}
    \centering
    \includegraphics[width=0.9\linewidth]{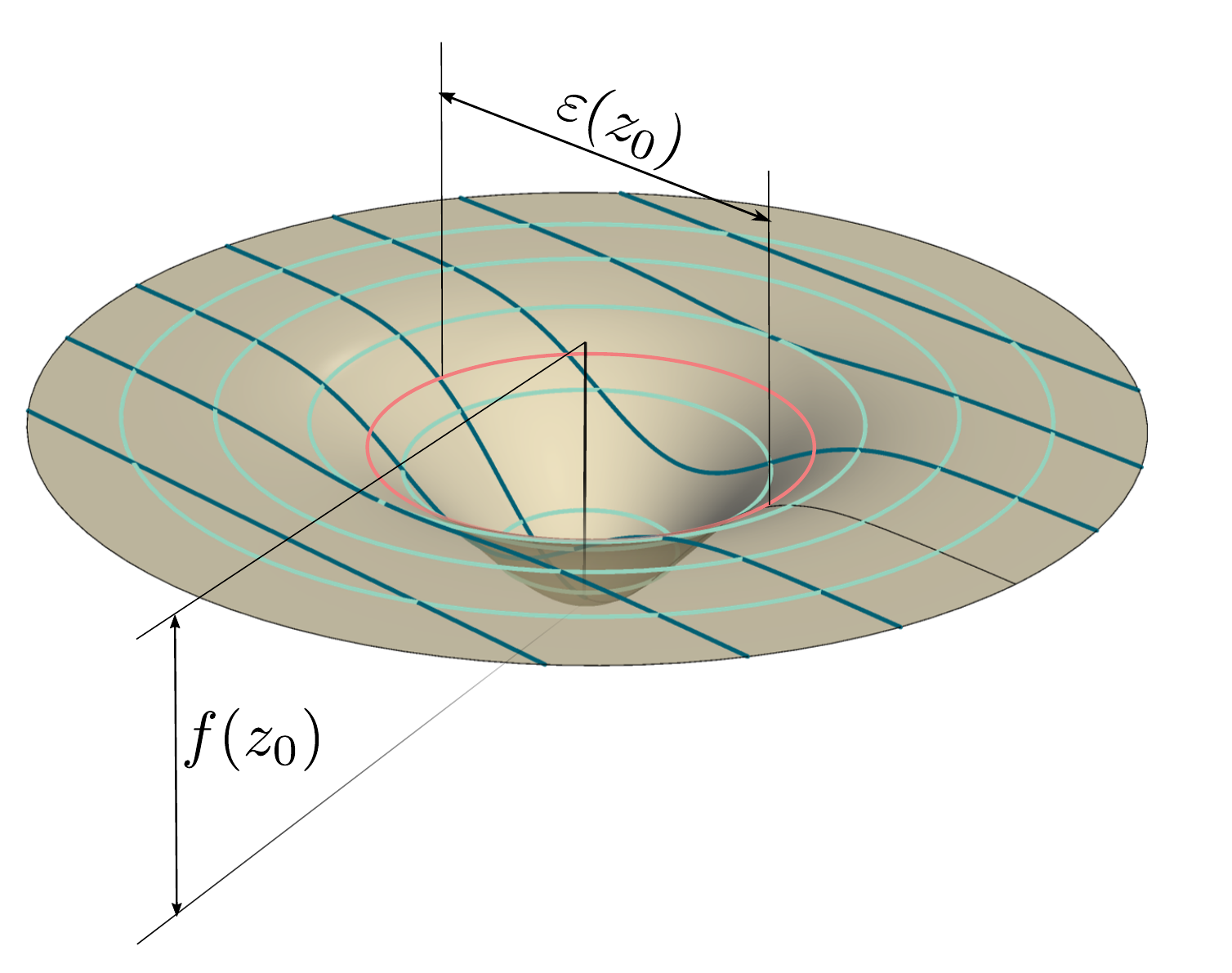}
    \caption{Schematic visualization of a valley (troughs in higher dimensions not indicated):
    Driving with $i C$ (cyan lines), does not change the value of the functional.
    Driving with $i B$ is in some sense the `orthogonal' direction.
    In a valley, all $i B$ trajectories have a local minimum.
    We identify the size of a valley by the region in which the second derivative is positive.}
    \label{figure:Bowl}
\end{figure}

In the case of the QAOA an intuitive explanation for this can be given.
We can regard the trough as a line of local minima, as any $i B$ movement away from it will increase the value of $F$ and, thus, be suppressed by a heuristic seeking local optimality.
By allowing for a component in $i C$ direction during each optimization step, we effectively increase the search routine's local dimensionality.
In general, the resulting step may not align with the true descent direction of the functional causing the next iterate to be slightly displaced from the trough.
One can grasp the situation by imagining a marble placed marginally off the trough depicted on the left in \autoref{figure:Troughs}.
While oscillating around the trough, it will gradually move toward the local minimum.
A common problem in such a situation is that, despite many steps being performed, the value of the functional improves very slowly, if at all.
The severity of this issue depends on the 'steepness' the trough.
Although, refined routines like the conjugate gradient method tend to avoid this behavior and adapting them for our special situation in the QAOA seems plausible, there are points within a trough where local search routines have to face further obstacles.

\begin{observation}\label{observation:LocalTraps}
    Eigenstates of $H$ that are also in $\tb$ act as local attractors; thus, local search tends to get stuck in these.
\end{observation}
We will mark those points as \emph{traps}.
Recall from \autoref{section:AsymptoticCircuits} that all non-optimal eigenstates of $H$ correspond to saddle points in the landscape.
Here, first derivatives vanish such that those states are located in $\tb$ whenever their second derivative in $i B$ direction is positive.
Eigenstates of $H$ are also eigenstates of $C$ rendering them invariant under $i C$ movements.
Moreover, any point in the state space will maintain its distance from an eigenstate of $H$ under $i C$ movements.
We can see this by checking that the trace distance (which is unitarily invariant) between a state $\phi$ and an eigenstate $\ketbra{z_{0}}$ stays constant along a $i C$ trajectory, i.e.\ we have 
\begin{align}
    \Vert \pi_{\phi}(e^{i \gamma C}) - \ketbra{z_{0}}\Vert_{1}
    &= \Vert e^{- i \gamma C} \phi e^{i \gamma C} - \ketbra{z_{0}}\Vert_{1}\nonumber\\
    &= \Vert\phi - e^{i \gamma C} \ketbra{z_{0}} e^{- i \gamma C}\Vert_{1}\nonumber\\
    &=\Vert\phi - \ketbra{z_{0}}\Vert_{1}\ .\label{equation:ConstantTraceDistance}
\end{align}
Due to the saddle point property, the functional will be almost constant, and no $i C$ step allows to leave the region.

Let $\ketbra{z_{0}}$ be an eigenstate of $H$ with eigenvalue $f(z_{0})$. In order to decide whether this state is in $\tb$ we consider its second derivative in $i B$ direction. 
Using the computation in \eqref{equation:SecondDerivative} we obtain
\begin{align}
    \nabla_{i B}^{(2)} F \big(\hspace*{-2pt}\ketbra{z_{0}}\hspace*{-2pt}\big)
    &= \tr\left( \ketbra{z_{0}} F_{2}^{B} \right)\nonumber \\
    &= 2 \sum_{\Ham(z_{0}, z) = 1} f(z) - N f(z_{0}),\label{equation:SecondDerivativeB}
\end{align}
where $\Ham(z, z')$ denotes the Hamming distance between the two bit strings $z$ and $z'$.
This expression has a nice interpretation.
In the above, the sum 
\begin{align*}
	\sum_{\Ham(z_{0}, z) = 1} f(z) = \bra{z_{0}}BHB\ket{z_{0}}
\end{align*}
comes from the special form of $B$ that is characteristic for the QAOA.
Up to a factor of $N$, it can be interpreted as the average of $f$ taken over all next nearest neighbors of $z_{0}$ in the hypercube $\bits$.
For what follows, it is useful to define the \emph{discrete gradient mean} at a point $z_{0}$,
\begin{align}
	\smv(z_{0}) \coloneqq \sum_{\Ham(z_{0}, z) = 1} \frac{f(z)- f(z_{0})}{N},
    \label{equation:Mu}
\end{align}
i.e., the average difference between the value of $f$ at $z_{0}$ and all neighboring strings.
This average is proportional to the second derivative and suffices to determine what happens to the saddle point at $z_{0}$ when restricting to the deep-circuit regime. 

Moreover, note that a string $z_{0}$ of length $N$ only has $N$ neighbors in $\bits$.
Therefore, the quantity $\mu(z_{0})$ can be efficiently determined by a classical computation that merely checks and adds the values of $f$ at those points.

\subsection{\label{subsection:TrapSizes}Trap Sizes}

Once a local search gets stuck in a trap, jumps, i.e. steps on a larger scale, have to be performed in order to escape.
We can estimate the required jump size by considering the distance from a trap to the closest state with negative second derivative.
This, arguably heuristic, quantity provides the scale at which the local behavior around a trap stops to dominate.
See \autoref{figure:Bowl} for a visualization.
From now on, we will refer to a region around a trap in which all second derivatives are non-negative as \emph{valley}. 

Assume a trap centered at a state $\ketbra{z_{0}}$ and an $\varepsilon$-neighborhood 
\begin{align*}
    \{\rho_\varepsilon \in \Omega \defcolon \norm{\rho_\varepsilon - \ketbra{z_{0}} {z_{0}}}_{1} < \varepsilon\}.
\end{align*} 
By definition we can express any state in this neighborhood as
\begin{align*}
    \rho_\varepsilon = \ketbra{z_{0}}{z_{0}} + \varepsilon K \text{ with } \norm{K}_{1} \leq 1, \ K = K^\dagger.
\end{align*}

Thus, we can bound the second derivative of $\rho_{\varepsilon}$ by the estimate 
\begin{align}
    \nabla_{i B}^{(2)} F\left(\rho_{\varepsilon}\right) 
    &=\tr\left(\rho_{\varepsilon} F_{2}^{B} \right) \nonumber\\
    &\geq \inf_{\Vert K\Vert_{1}\leq 1} 
    \left(\tr\left(\ketbra{z_{0}} {z_{0}} F_{2}^{B} \right) + \varepsilon\tr\left(K F_{2}^{B} \right)\right) \nonumber\\
    &= 2 N \mu(z_{0}) -\varepsilon \sup_{\Vert K\Vert_{1}\leq 1} \tr\left(K F_{2}^{B} \right) \nonumber\\
    &= 2 N \mu(z_{0}) - \varepsilon\Vert F_{2}^{B}\Vert_{\infty}   \label{equation:EstimateDeep}
\end{align}
where we used the variational formula for the operator norm in the last step.
Thus, by determining or bounding the universal factor $\Vert F^{B}_{2}\Vert_{\infty}$ for an explicit problem instance, one can estimate the size of a valley in the state space landscape by identifying all $\varepsilon$-neighborhoods where the second derivatives are positive.

First, the last line of \eqref{equation:EstimateDeep} bounds the critical $\varepsilon$; i.e.,
\begin{align}\label{equation:EpsilonEstimate}
    \varepsilon <  \frac{2 N \smv (z_{0})}{\Vert F_{2}^{B}\Vert_{\infty}}.
\end{align}

Second, we can provide an upper bound for $\Vert F_{2}^{B} \Vert_{\infty} $ by exploiting the sub-multiplicativity of the operator norm and eventually obtain
\begin{align}\label{equation:AlphaEstimate}
\Vert F_{2}^{B} \Vert_{\infty} &= \Vert [B, [B, C]] \Vert_{\infty} \nonumber\\
&\leq 4 \Vert B \Vert^{2}_{\infty} \Vert C \Vert_{\infty} = 4 N^{2} \Vert C \Vert_{\infty}.
\end{align}
Combined, we conclude 
\begin{align}\label{equation:FinalEpsilonEstimate}
\varepsilon < \frac{\mu(z_{0})}{2 N \Vert C \Vert_{\infty}}.
\end{align}
This estimate naturally suggests that the region around an eigenstate - particularly the global minimum - where the $B$ derivative exhibits the same properties as the eigenstate itself, decreases in size as the problem size increases ($\sim 1 / N$).
Furthermore, we can incorporate the operator norm of $C$ into $\mu$, by noting that within $\mu$ we can always consider eigenvalues normalized with respect to the largest one.
This yields a convenient redefinition
\begin{align}\label{equation:MuNew}
    \Tilde{\mu}(z_0) \coloneqq \sum_{\Ham(z_{0}, z) = 1} \frac{f(z) - f(z_{0})}{N \Vert C\Vert_\infty}
\end{align}
such that 
\begin{align}\label{equation:FinalEpsilonEstimateNew}
    \varepsilon < \frac{\Tilde{\mu}(z_{0})}{2 N}.
\end{align}
Therefore, the properties of local minima with respect to $B$ depend only on the relative distances between the eigenvalues of $H$ as well as on the problem size $N$.

\begin{figure*}[!ht]
    \begin{minipage}[b]{0.32\linewidth}
    \flushleft (a) \\ \centering
    \includegraphics[width=0.83\linewidth]{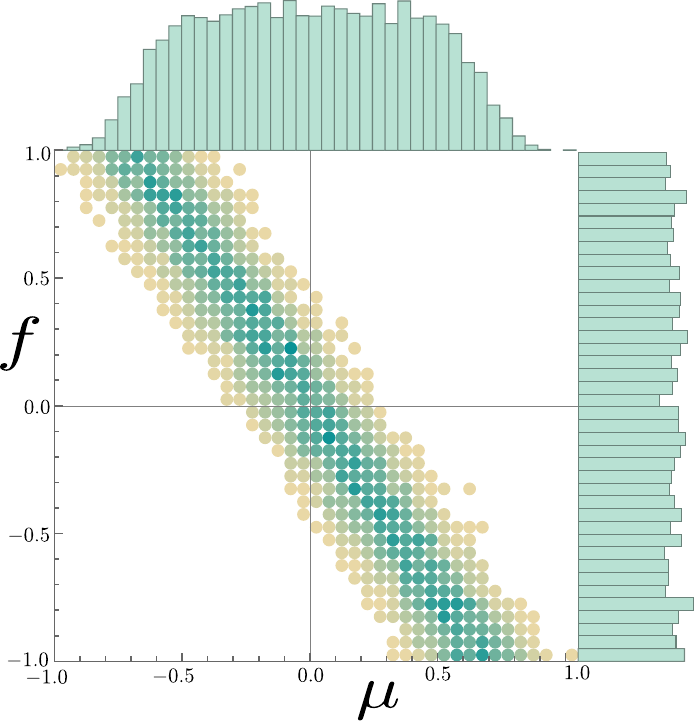} 
    \end{minipage}\hfill
    \begin{minipage}[b]{0.32\linewidth}
    \flushleft (b) \\ \centering
    \includegraphics[width=0.9\linewidth]{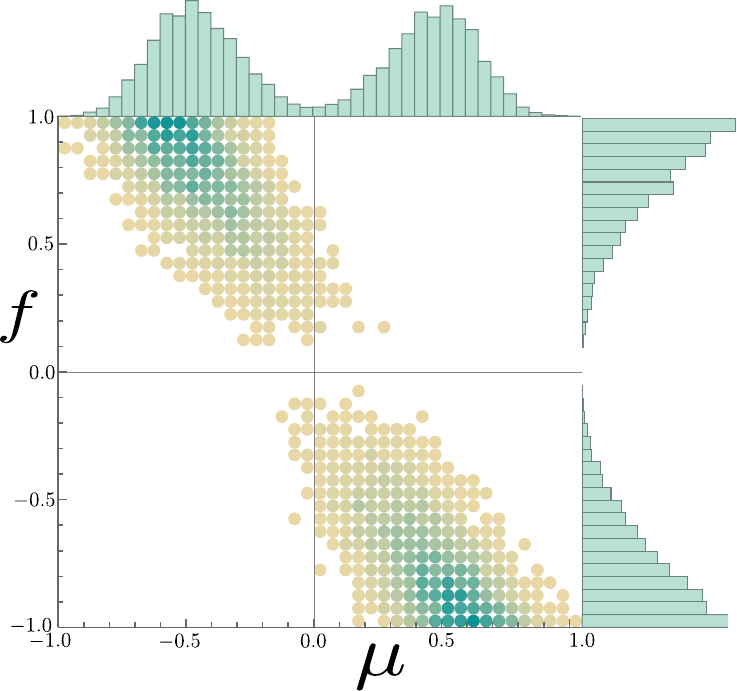}
    \end{minipage}
    \begin{minipage}[b]{0.32\linewidth}
    \flushleft (c) \\ \centering
    \includegraphics[width=0.87\linewidth]{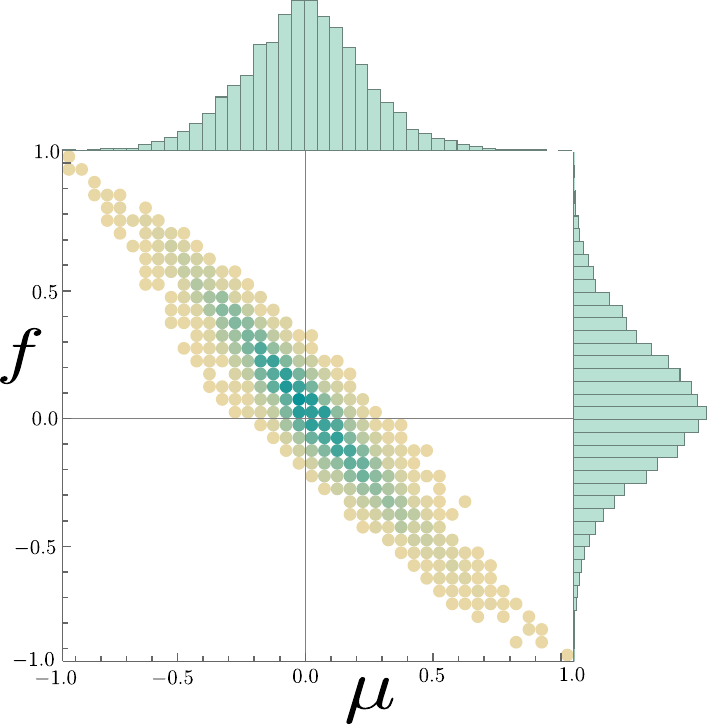}
    \end{minipage}
    \caption{$\mu$\,-$f$ diagrams of some randomly sampled functions $f$ on $13$ bits.
    The values of $f$ are (a) uniformly randomly distributed with support on $[- 1, 1]$ and (b) distributed with respect to a bimodal distribution that favors values at the boundaries of $[- 1, 1]$. From a local search perspective these are considered unfavorable instances. (c), on the contrary, stems from a random QUBO instance. This diagram unveils an optimization landscape that is favorable for QAOA.}
    \label{figure:MuFRandom}
\end{figure*}

\subsection{\label{subsection:APerformanceIndicator}A Performance Indicator}

Valleys have a critical influence on the performance of deep-circuit QAOA as they are encompassing attractors and traps for local search routines.
Intuitively, a landscape with 'too many' valleys, i.e. too many traps, seems to be unfavorable for a good performance. The previous section now equips us with tools for making such a statement more refined:
In the following we consider the distribution of valleys with respect to their number, size, and depth. 
Then, the statistics of these properties allow us to identify obstacles for local search routines. By this we obtain an accessible performance indicator for deep-circuit QAOA.

Recall that the presence of a valley can be attested by the discrete gradient mean $\mu(z)$, where  $\mu(z) \geq 0$ indicates a valley with a radius (in trace distance) bounded by \eqref{equation:FinalEpsilonEstimateNew}. 
For a given target function $f : \bits \rightarrow \R$, consider the set 
\begin{align}\label{equation:MuFSet}
    \Xi_{f} = \{(f(z), \mu(z))\, |\, z \in \bits\}.
\end{align}
From now on, we refer to a density plot of $\Xi_{f}$ as the $\mu$\,-$f$ diagram of $f$.
A significant amount of the optimization landscape's quantitative structure in state space can be directly observed from such a diagram.
Several examples are given in \autoref{subsection:Examples}.

Relevant questions that can be directly answered by examining the $\mu$\,-$f$ diagram are:

\begin{itemize}
    \item[(i)] What is the fraction of points with $\mu > 0$?
    \item[(ii)] Are there correlations between the radius of a valley and its depth? Are the largest valleys also the deepest?
    \item[(iii)] Is there a separation between small and large valleys? Are large valleys less likely than small ones?
\end{itemize}

Quantitative answers to these questions allow us to estimate whether a landscape is in principle favorable or unfavorable for local search routines.

Generally, a landscape with only few valleys (i) can be considered favorable.
Non-primitive routines for local search include subroutines for jumps that allow to escape a local trap and search for an optimum elsewhere.
This is, however, only efficient if there are not too many traps.
Without further structure given, we will assume that valleys with a large radius are more likely to attract a local search routine.
Landscapes whose largest valleys are also the deepest (ii) are favorable, as a local search is more likely to find the global minimum of $f$. 
A central aspect of fine-tuning a local search routine is adjusting the sizes of local steps and jumps.
Here, statistical information about the valley sizes (iii) can be very useful.
In a \emph{favorable landscape}, most valleys are small and shallow and only few valleys (including those we are looking for) are large and deep.
Taking step and jump sizes just large enough will then allow us to set up a local search that effectively ignores small valleys and is only attracted to large ones.
In conclusion, we will consider landscapes as favorable if the largest valley is also the deepest and the distribution of $\mu$ and $f$ thins out for increasing $\mu$ and decreasing $f$, i.e. to the low-right corner of the $\mu$\,-$f$ diagram. 
In the sense that it indicates the presence of performance obstacles, we regard the $\mu$\,-$f$ diagram of a target objective $f$ as a two dimensional performance indicator. 

Here we emphasize once more that the $\mu$\,-$f$ diagram relies solely on data obtainable from classical, efficiently accessible properties of $f$ on strings $z$; specifically, the value of $f$ and its average over nearest neighbors.
On a more refined level, assessing the unfavorability of a landscape involves estimating a tail distribution.
This task is central to the mathematical field of risk management \cite{Baker2022WassersteinSolutionQualityAndTheQuantumApproximateOptimizationAlgorithmAProtfolioOptimizationCaseStudy} and many statistical methods developed there can be applied to our problem.
At this point, however, we leave a detailed analysis and application of those methods for future work and restrict to the provision of examples.

\section{\label{section:NumericalResults}Numerical Results}

In the previous sections, we explored various aspects of the interplay between local search routines and the geometry of the state space optimization landscape.
In this section, we will employ a simple local search routine to demonstrate these on concrete examples. 

At this stage, we will obtain our results by numerically simulating the behavior of a noise-free quantum computer.
Observations inferred from this clearly over-idealized setting therefore only account for ruling out the feasibility of a problem instance and not for ultimately demonstrating it. 

Rather than using Qiskit or similar existing frameworks, we perform simulations on the level of explicitly implementing the matrices and vectors involved.
By this, all major computations can be directly performed by matrix-vector multiplications in an optimized library, such as BLAS.
We thus ensure that the runtime of our simulations scales linearly with the circuit depth $p$, enabling us to simulate very deep circuits within fractions of a second.

\autoref{subsection:Examples} presents a diverse set of examples along with their $\mu$\,-$f$ diagram used to examine the local search algorithm presented in \autoref{subsection:ABasicLocalSearchAlgorithm}.
Indicators for assessing the algorithm's performance are introduced in \autoref{subsection:DescriptionOfDepictedDate} and analyzed in \autoref{subsection:PerformanceOfDeepCircuitLocalSearch}.
There, we also demonstrate the significance of our newly introduced performance indicator, the $\mu$\,-$f$ diagram, using the examples.
Finally, the observations discussed in \autoref{section:DeepCircuits} are detailed in \autoref{subsection:LocalTraps} and \autoref{subsection:StepSizeAnalysis}.

\begin{figure*}[!htp]
    \centering
    \includegraphics[width=0.7\textwidth]{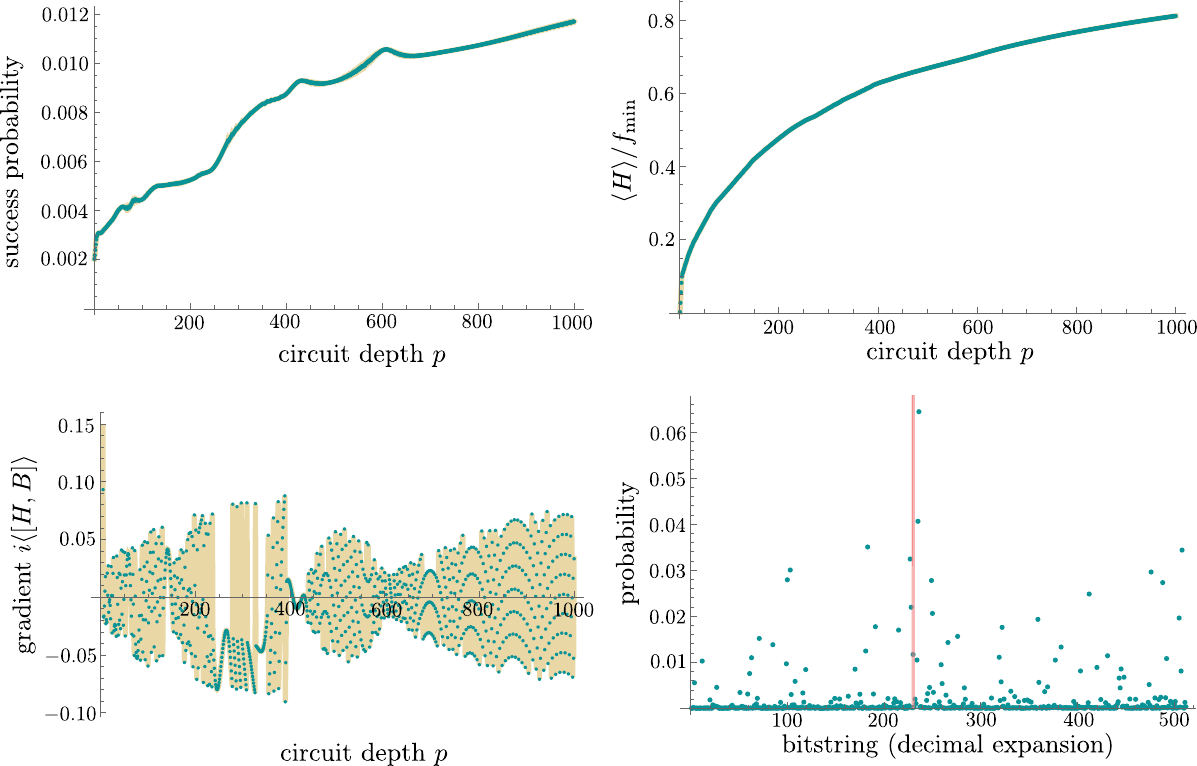}
    \caption{Minimization of a function with uniformly randomly distributed values (step size $\varepsilon = 0.1$).
    As the $\mu$\,-$f$ diagram (\autoref{figure:MuFRandom}(a)) suggests, layer-wise optimization struggles to find an optimum.
    Both, the success probability and the approximation ratio increase only very slowly, and the final state spreads over many computational basis states, with the focus on a non optimal state.
    }
    \label{figure:Uniform}
\end{figure*}

\subsection{\label{subsection:Examples}Examples}

A central question in assessing the potential of QAOA is identifying problem classes and instances where the algorithm is likely to perform well.
However, one must anticipate the existence of 'no free lunch'- theorems implying that most instances will not perform well.   

We immediately observe such a behavior when considering objectives $f$ in which the values of $f$ are distributed without paying attention to the topology of $\bits$.
Examples are given in \autoref{figure:MuFRandom}.
In each case, we first generated the values of f according to a specified distribution function and then randomly assigned them to bit strings using uniform sampling.
In the first example, $f$ is uniformly distributed over the interval $[-1, 1]$, resulting in an unfavorable landscape in which the valley sizes are uniformly distributed, as well.
In the second example (also shown in \autoref{figure:MuFRandom}, the values of $f$ follow a bi-modal distribution that favors both large and small values, leading to a more structured landscape.
This illustrates a clear instance of the "no free lunch" behavior, as the global minimum does not possess any statistically distinguishing features relative to many of the local traps.

This situation changes when we consider random instances of QUBO, i.e.\ functions of the form
\begin{align*}
    f(z) = \sum_{i, j} z(i) M_{i j} z(j),
\end{align*}
generated from a randomly chosen hermitian matrix $M$ with unit norm. 
Here, we observe a mutual dependence among the values of neighboring bit strings, accompanied by a thin-tailed distribution -- clearly visible on the right in  \autoref{figure:MuFRandom} -- which results in a favorable landscape.
In correspondence, we also see from our numerical studies will reveal that even the a local search using the simple downhill simplex method performs better on these instances. We expect that this tail behavior can be explained by employing the central limit theorem.
However, we leave a full proof for future work.

A review of the catalog of of special classes of pseudo-Boolean functions \cite{Crama2011BooleanFunctionsTheoryAlgorithmsAndApplications}  quickly reveals additional classes that admit favorable optimization landscapes.
Among them are functions that are zero for the vast majority of bit strings, with negative values occurring only on a small subset.
Here, the landscape is mostly flat, and all valleys will encompass a global minimum. 

For example, SAT problems with only few feasible points constitute such functions. 
Remarkably those are in close correspondence to the `find the marked entry in a database' problem Grover's search algorithm \cite{Grover1996AFastQuantumMechanicalAlgorithmsForDatabaseSearch} solves with a proven quantum speedup.

\begin{figure*}[!htp]
    \centering
    \includegraphics[width=0.7\textwidth]{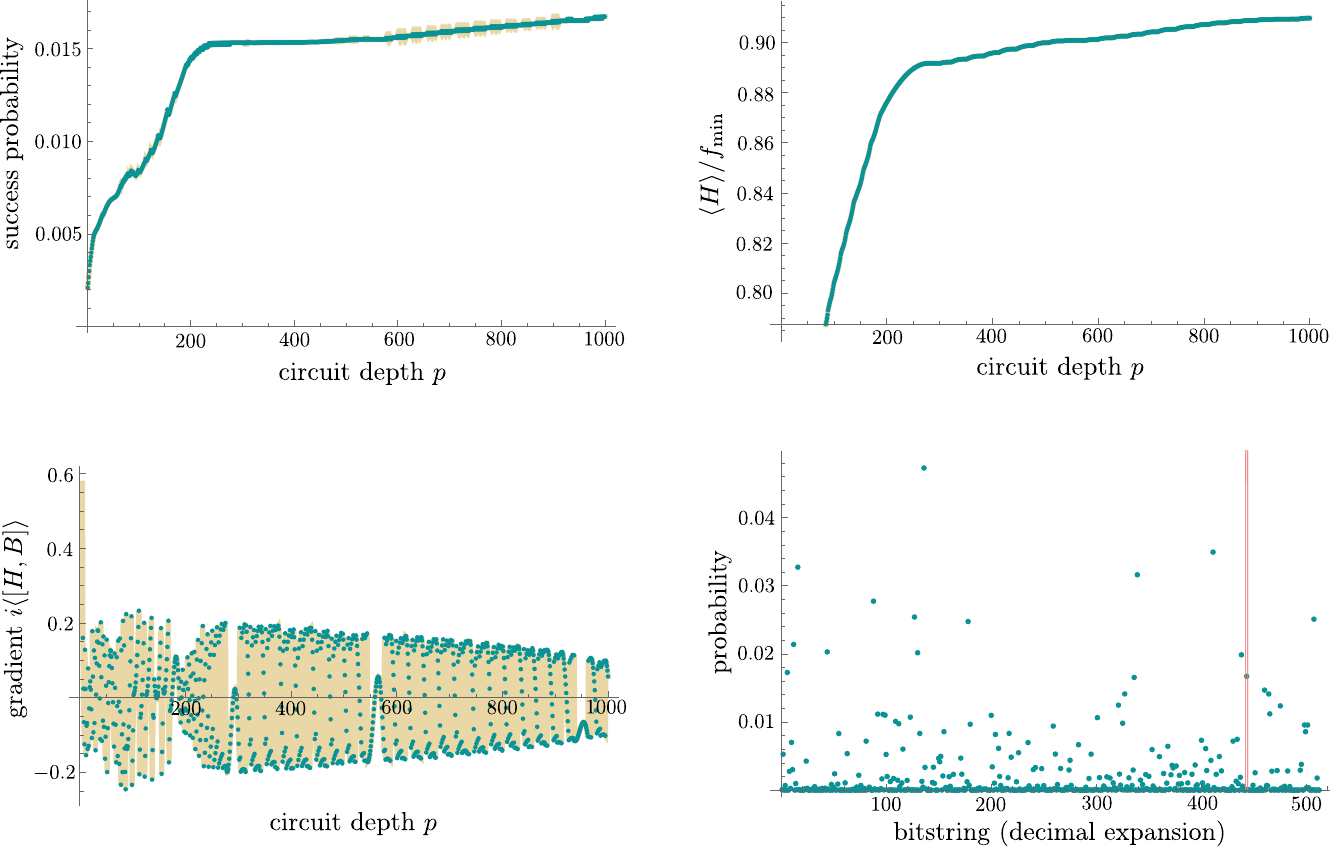}
        \caption{Minimization of a function with bimodally randomly distributed values (step size $\varepsilon = 0.1$).
        In alignment with the prediction of the $\mu$\,-$f$ diagram (\autoref{figure:MuFRandom}(b)), layer-wise optimization performs poorly on this instance.
        After success probability and approximation ratio quickly increase within the first 200 iterations, they plateau and the final state's outcome distribution has no clear focal point.
        }
    \label{figure:Bimodal}
\end{figure*}

\subsection{\label{subsection:ABasicLocalSearchAlgorithm}A Basic Local Search Algorithm}

We use a basic and naive local search routine that may be viewed as a state-space version of \emph{pattern search} \cite{Hooke1961DirectSearchSolutionOfNumericalAndStatisticalProblems} or \emph{random search} \cite{Anderson1953RecentAdvancesInFindingBestOperatingConditions,Brooks1958ADiscussionOfRandomMethodsForSeekingMaxima}.
It works as follows:
\begin{itemize}
    \item[0.] Fix an initial state $\rho_{0}$ and a set of $m$ unitaries $\mathcal{U}^{\varepsilon} = \{U^\varepsilon_{1}, \dots, U^{\varepsilon}_{m}\}$ and set $p = 0$.
    \item[1.] For $i = 1, \dots, m$, prepare states $\tau_{i} = U^{\varepsilon}_{i} \rho_{p} U^{\varepsilon}_{i}\phantom{}^{*}$.
    \item[2.] Measure the target Hamiltonian on the states $\tau_{i}$, i.e. compute $f_{i} = F(\tau_{i})$
    \item[3.] Pick $i^*=\operatorname{argmin}_{i = 1,\dots,m} f_{i}$, set $f_{i^{*}}$ as active estimate, set $\rho_{p + 1}=\tau_{i^{*}}$ and go to $1$.
\end{itemize}
If all $U^{\varepsilon}_{i}$ are close to the identity this will generate a simple local search routine.
Here, the step size in state space is determined by the difference between $U^{\varepsilon}_{i}$ and the identity, evaluated on the active state $\rho_{p}$.
An upper bound to this distance is given by the maximum operator norm distance between the $U^\varepsilon_{i}$ and the identity. 
In our concrete examples, we take the $U^\varepsilon_{i}$ to be a single layer of QAOA unitaries, with parameter tuples $(\beta_{i},\gamma_{i})$ chosen from a grid-like pattern centered around zero.
The grid spans a range of $\varepsilon/ \Vert H \Vert_{\infty}$, with $11$ equidistant values for $\beta$ and $5$ for $\gamma$.
For small $\varepsilon$, and by setting $\Vert H \Vert_{\infty} = 1$, the parameter $\varepsilon$ corresponds, at least approximately, to an upper bound on the actual step size in state space.
In any case, $\varepsilon$ can be regarded as an effective size parameter that defines the characteristic scale on which the local search routine operates. 

It is important to point out that our intention for employing this type of algorithm is to investigate the fundamental behavior of local search routines and not to achieve the best possible performance.
They usually offer high robustness in progressively improving estimates of an optimization target, without requiring extensive preconditioning tailored to a specific problem instance.
This robustness, however, typically comes at the price of requiring many optimization rounds.
To achieve faster convergence, more refined algorithms and considerable case-specific feature engineering is necessary.

\begin{figure*}[!htp]
    \centering
    \includegraphics[width=0.7\textwidth]{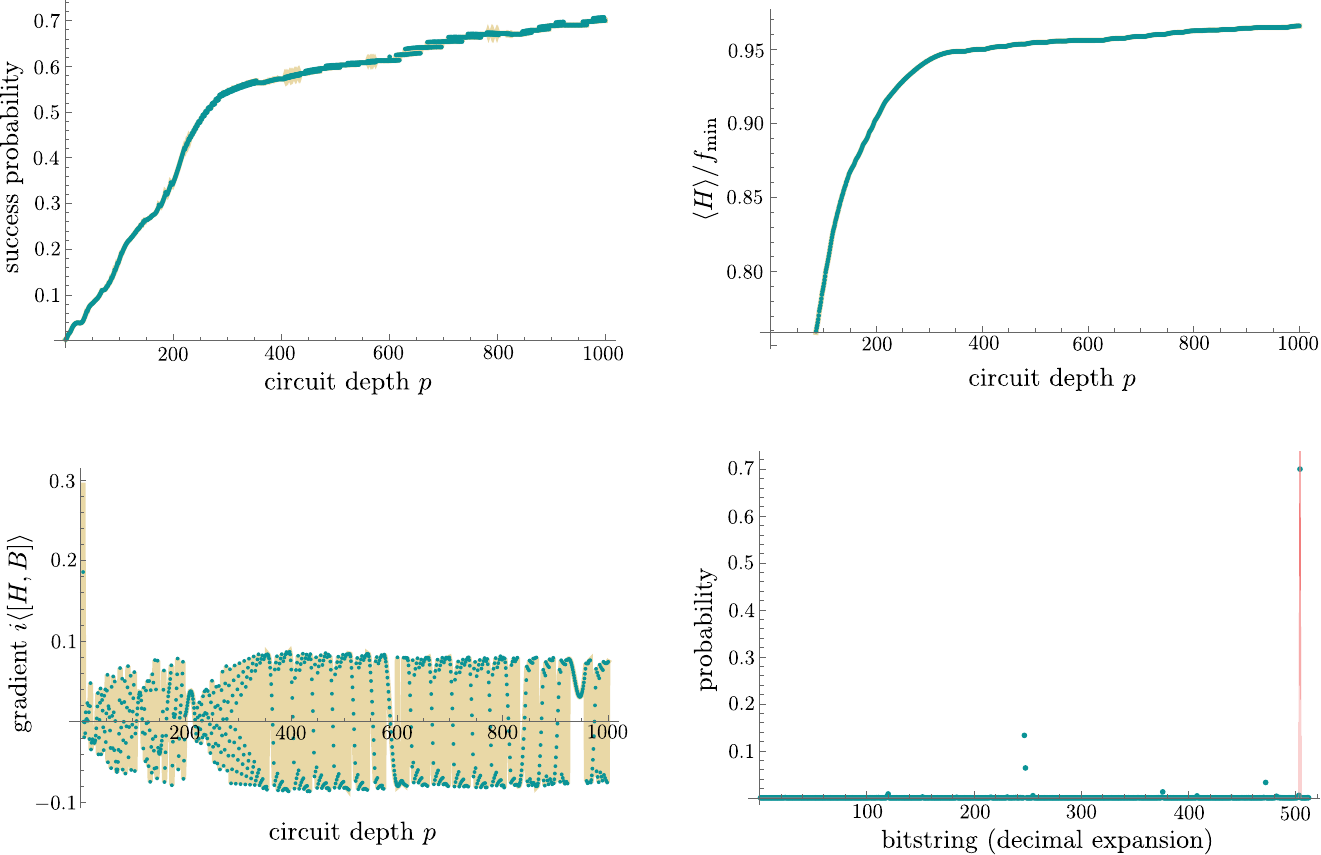}
    \caption{Minimization of a random QUBO instance on nine bits (step size $\varepsilon = 0.1$).
    In accordance with the $\mu$\,-$f$ diagram(\autoref{figure:MuFRandom}(c)), the layer-wise optimization reliably finds the global optimum.
    Both, success probability and approximation ratio increase significantly with circuit depth, and the final state exhibits a strong concentration around the optimum with only a small amplitude spread over other computational basis states.}
    \label{figure:RandomQUBO}
\end{figure*}

\subsection{\label{subsection:DescriptionOfDepictedDate}Description of Depicted Data}

Characteristic information from each run is presented through the following indicators:
\begin{itemize}
    \item[(i)] \textbf{Success probability} of obtaining the optimal outcome string for a given optimization problem. This is computed by measuring the overlap of the optimal state with the  respective active state $\rho_{p}$, and plotted as a function of the circuit depth $p$, i.e., the number of algorithmic rounds.
    \item[(ii)] \textbf{Approximation ratio} $\langle H\rangle /f_{\min}$ obtained by measuring the target Hamiltonian on the active state $\rho_{p}$; again as a function of the number of optimization rounds.
    \item[(iii)] \textbf{Gradient magnitude} of the functional $F$ along an $i B$ trajectory, evaluated at the active state $\rho_{p}$; again as a function of the number of optimization rounds.
    \item[(iv)] \textbf{Outcome probability distribution} of bit strings obtained by measuring the final active state $\rho_{p_{\max}}$ in the computational basis.
    Bit strings are mapped to integers via binary-to-decimal conversion.
    In each plot, the red line indicates the position of the optimal solution(s) for the given problem instance.
\end{itemize}

\begin{figure*}[!htp]
    \centering
    \includegraphics[width=0.7\textwidth]{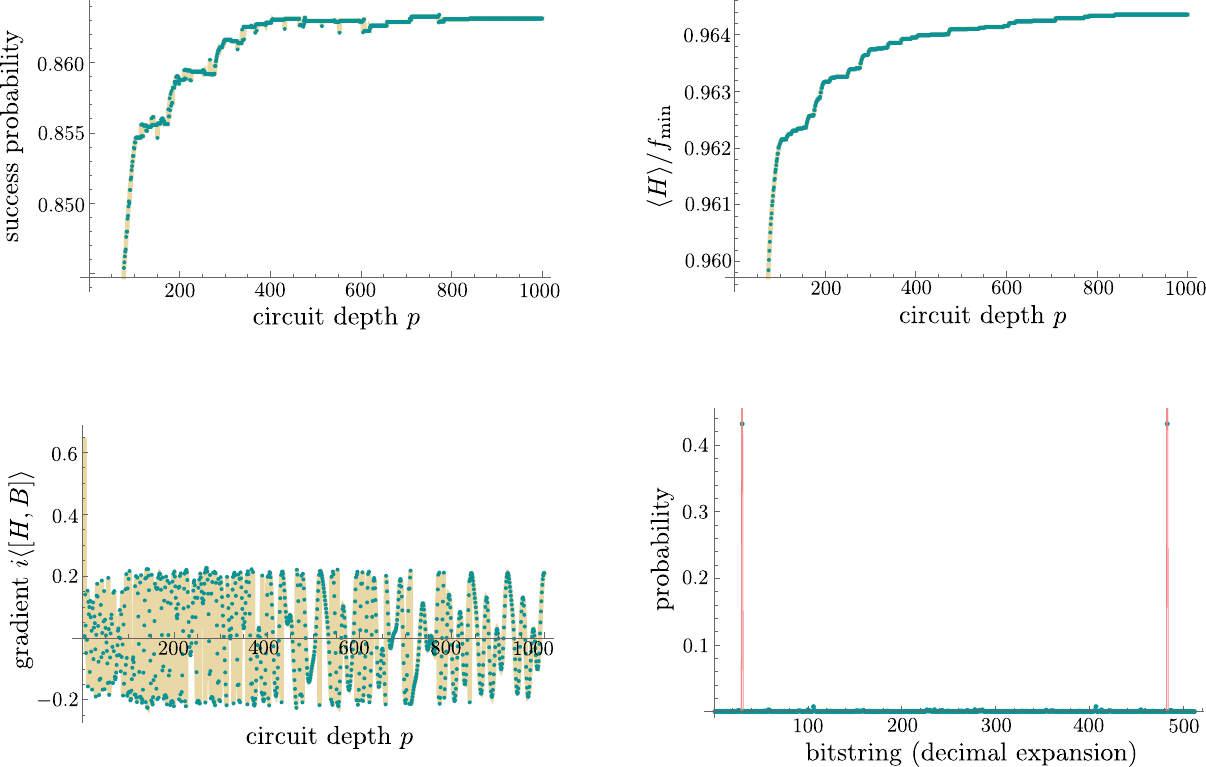}
    \caption{Solving MAXCUT on random graph (step size $\varepsilon = 0.1$).
    Since MAXCUT is formulated as a QUBO problem, a $\mu$\,-$f$ diagram similar to \autoref{figure:MuFRandom}(c) is expected, which again provides an accurate prediction:
    The layer-wise optimization is able to find the two optima with almost certainty.
    Success probability and approximation ratio rapidly increase to practical levels within the first 100 iterations.
    The final state exhibits an almost perfect concentration around the global optima.}
    \label{figure:RandomGraph}
\end{figure*}

\subsection{\label{subsection:PerformanceOfDeepCircuitLocalSearch}Performance of Deep-Circuit Local Search} 

In this subsection, we analyze the results gathered from the local search routine with step size $\varepsilon = 0.1$ on example instances with nine bits, in contrast to the 13-bit instances used in the $\mu$-$f$ diagrams depicted in \autoref{figure:MuFRandom}.
However, we will observe, qualitatively, the exact performance behavior that is predicted by the $\mu$\,-$f$ diagrams for the larger instances.
First, the performance on the objective function with random values sampled from a uniform distribution on $[-1, 1]$ depicted in \autoref{figure:Uniform}:
Although the approximation ratio increases steadily, the success probability of approximately $0.0125$ after 1000 iterations is comparatively small.
Accordingly, the final outcome distribution is spread out, showing no significant peak around either local or global minima.
In fact, the approximation ratio seems to slowly increase proportionally to $p^{2}$.
Furthermore, the gradient oscillates as a function of $p$ with slowly oscillating amplitude.

Second, we analyze a concrete example of a random function with values sampled from a bimodal distribution on $[-1, 1]$ (see \autoref{figure:Bimodal}):
At first glance, the most striking feature is the rapid increase of the approximation ration within the first 300 iterations, followed by its saturation at a comparatively high level.
Since this behavior is not reflected in the success probability -- which remains at approximately $0.015$ even after 1000 iterations -- we infer that the algorithm initially populates states with good, but suboptimal, objective values.
This is expected, as functions of this kind typically possess a large number of such states.
In that spirit, the final outcome distribution is as spread out as in the uniform case and the gradient once again oscillates as a function of $p$, with a slowly decreasing amplitude.

\begin{figure*}[!htp]
    \centering
    \includegraphics[width=0.7\textwidth]{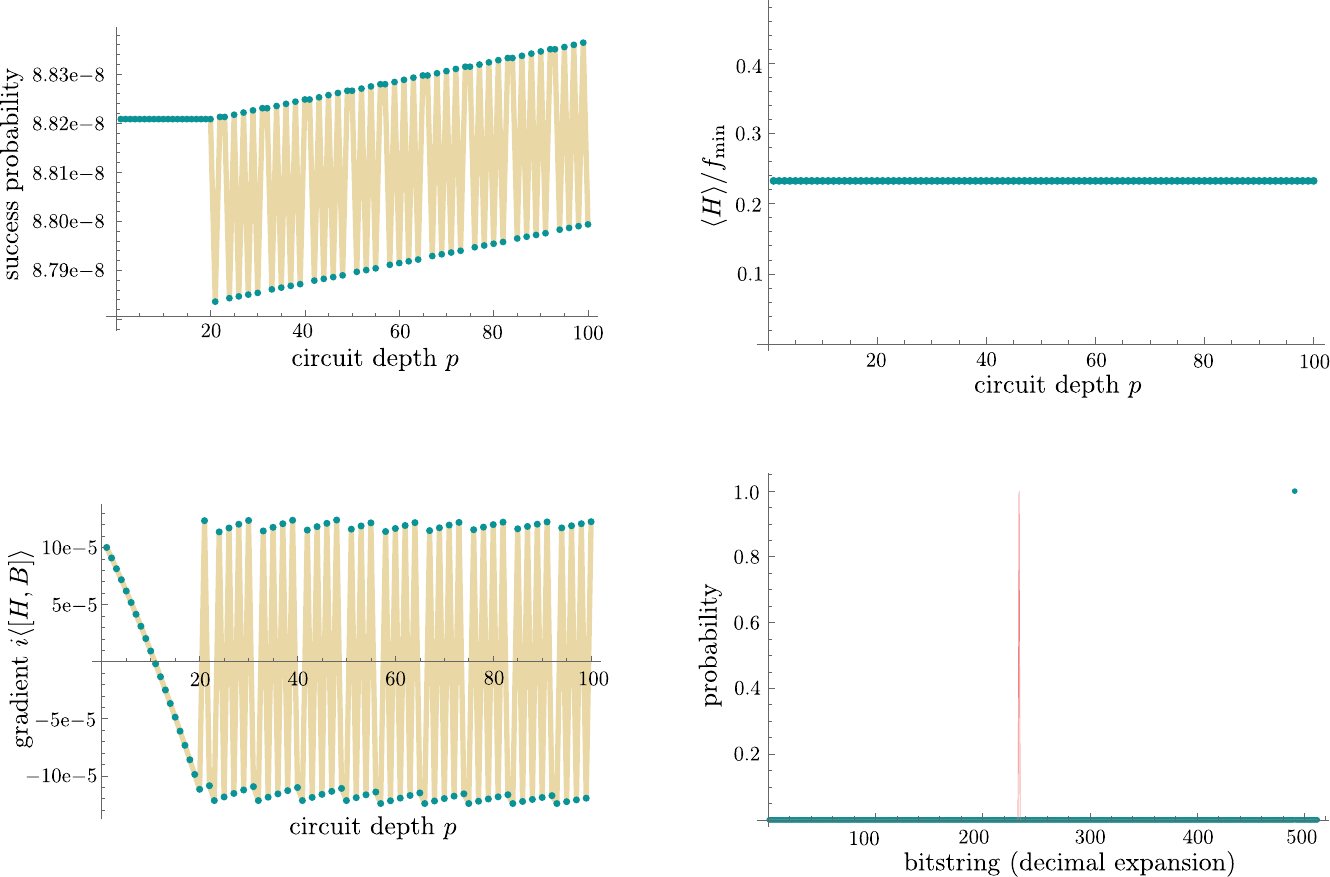} 
    \caption{Minimization of a random QUBO instance starting with a state $\ketbra{\phi_{0}}{\phi_{0}}$ that is close to a non-optimal state $\ketbra{z}{z}$ (step size $\varepsilon = 0.001$).
    Although the $\mu$\,-$f$ diagram predicts, in principle, a favorable optimization landscape, the chosen initial state as well as the very small step size do not allow for finding the global optimum at all.
    Accordingly, the success probability starts out very low and remains basically unchanged; any supposed upwards trend is more likely due to numerical instabilities.
    The approximation ratio also remains constant at the value of $f(z) / f_{\min}$.
    The gradient's norm is negligibly small, indicating that the optimization routine is actually stuck at the local optimum $\ketbra{z}{z}$, as is also clear from the final state's amplitude distribution.}
    \label{figure:AroundZero}
\end{figure*}

Next, we revisit the instances that appear, based on the $\mu$\,-$f$ diagram, to correspond to favorable optimization landscapes.
As shown in \autoref{figure:RandomQUBO}, the numerical data from the random QUBO problem offers clear and unequivocal support for this hypothesis.
Not only does the approximation ratio increase significantly within the first 300 iterations, but the success probability also rises to as high as $0.7$ over the course of 1000 iterations -- substantially higher than in both preceding examples.
Accordingly, the final outcome distribution strongly concentrates around the optimal solution.
The gradient also behaves similarly to the bimodal case, regarding both its oscillation and its amplitude.

In addition, we consider a more concrete example: a MAXCUT instance for a random graph with nine vertices and unweighted edges.
The resulting objective function is also of QUBO-type and we therefore expect a favorable optimization landscape.
However, the function merely is integer-valued.
Furthermore, it is symmetric under flipping each bit of its argument.
Therefore, the criteria of \autoref{theorem:UniversalGates} are not met and it is thus of particular interest whether our performance indicator, the $\mu$\,-$f$ diagram, still applies.
Indeed, our numerical results (step size $\varepsilon = 0.1$) indicate a favorable landscape (see \autoref{figure:RandomGraph}):
After $1000$ iterations, the success probability even exceeds that of the previous QUBO instance by roughly $0.17$.
Likewise, the final outcome distribution strongly concentrates around the two optimal solutions.
We again observe a saturation of the approximation ratio after 300 steps, but with an even higher saturated ratio than in the previous case.
The gradient also admits oscillations in $p$, but its overall amplitude is doubled in comparison to the previous case.

\begin{figure*}[!htp]
    \centering
    \includegraphics[width=0.7\textwidth]{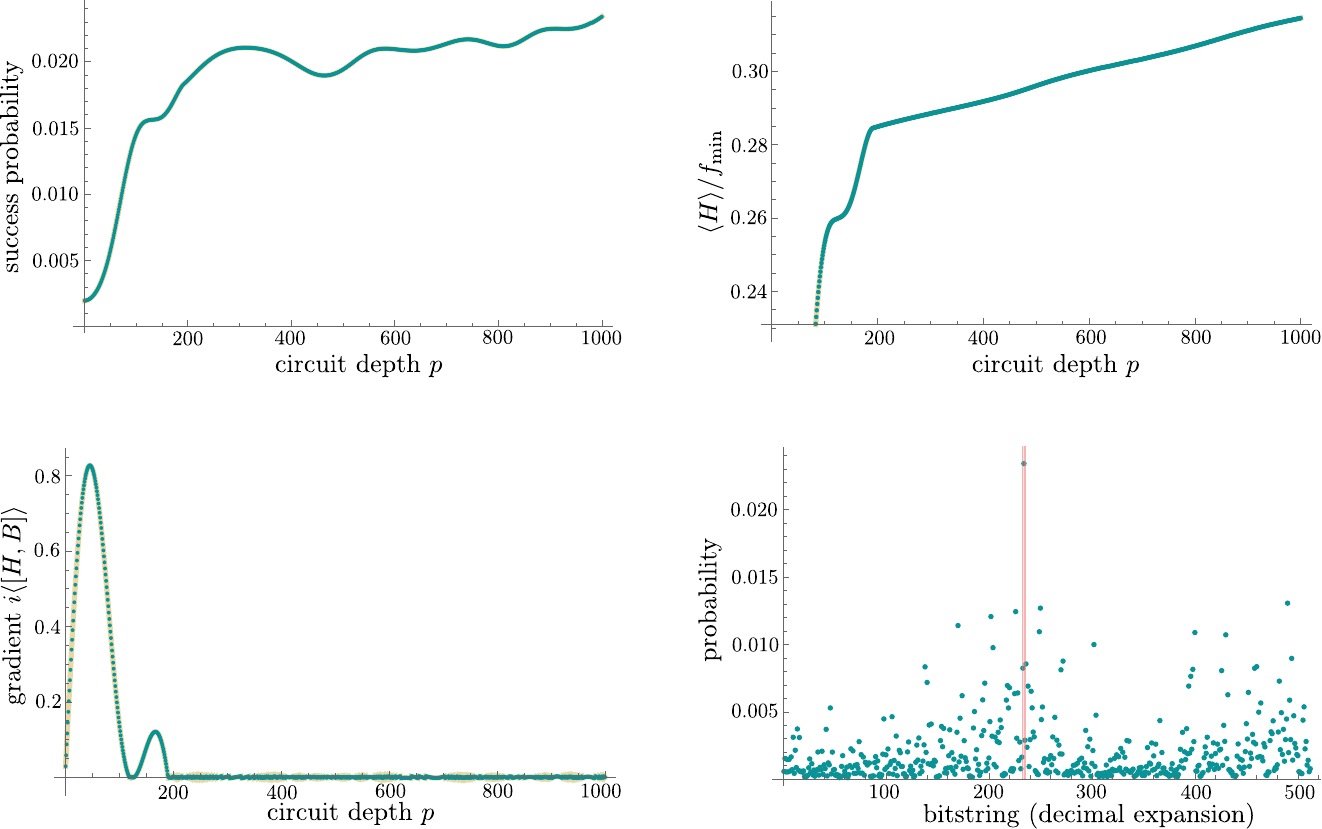}
    \caption{Minimization of a random QUBO function using a small step size ($\varepsilon = 0.01$).
    The success probability and the approximation ratio both show a slow, non-saturated increase with increasing circuit depth.
    The gradient's norm behaves similar to the random QUBO case with larger step size of $\varepsilon = 0.1$: starting with an initial peak and remaining small and almost constant afterwards.
    The final state has most of its amplitude spread out across many computational basis states.
    }
    \label{figure:SmallSteps}
\end{figure*}

\subsection{\label{subsection:LocalTraps}Troughs and Local Traps}

One particular feature that has only been briefly discussed so far is the behavior of the gradient across all problem instances.
Although not directly relevant for our algorithm -- since we are using a gradient-free optimization -- it nonetheless highlights a peculiarity of local search routines.
Starting from a point with a comparatively large gradient, the search typically reaches a region of moderate gradient magnitude within just a few steps.
This behavior corresponds to \autoref{observation:Trough}, where we identified such regions as troughs.
Beyond this point, the gradient begins to oscillate around zero, even as the local search continues to improve the objective function.
This suggests that the optimization steps are not aligned with the true descent direction of the objective, which led us to \autoref{observation:ZigZag}.

On the other hand, the existence of local traps plays a central role for the success chances of a local search routine.
As outlined in \autoref{subsection:LocalAttractorsAndTraps}, these traps arise from saddle points in the asymptotic landscape.
A low-layer optimization is unable to escape them, since the required move involves directions that demand unitaries of potentially asymptotic circuit depth.
In \autoref{observation:LocalTraps}, we identified them with the eigenstates of the target Hamiltonian.
\autoref{figure:AroundZero} provides strong numerical support for this.
There, we intentionally initialize the state as $\rho_{0} = \ketbra{\phi_{0}}{\phi_{0}}$, chosen to be close to a non-optimal eigenstate of a Hamiltonian corresponding to a random QUBO instance on 9 qubits.
We employed the algorithm from \autoref{subsection:ABasicLocalSearchAlgorithm}, using a small step size ($\varepsilon = 0.001$).
It was terminated after 100 iterations, as the results indicate that the search remains within the vicinity of the exact same non-optimal eigenstate.
The approximation ratio remains constant at the objective value of the non-optimal eigenstate, and the final probability distribution is entirely concentrated on this eigenstate.
A more refined behavior can be inferred from the success probability and the gradient which are both depicted on a very fine scale.
For the first 20 steps the success probability stays constant and the gradient slowly decreases.
Afterwards the gradient and the success probability begin to oscillate.
While the gradient stays small, the success probability begins to grow linearly.
Both effects happen on a very small scale, i.e., the gradient stays almost zero and the linear growth is almost flat.
Even though the effective step size was on order $10^{-3}$ the slope of the success probability is much smaller, namely on the order of $10^{-12}$.
Therefore, the algorithm is effectively trapped.

\begin{figure*}[!htp]
    \centering
    \includegraphics[width=0.7\textwidth]{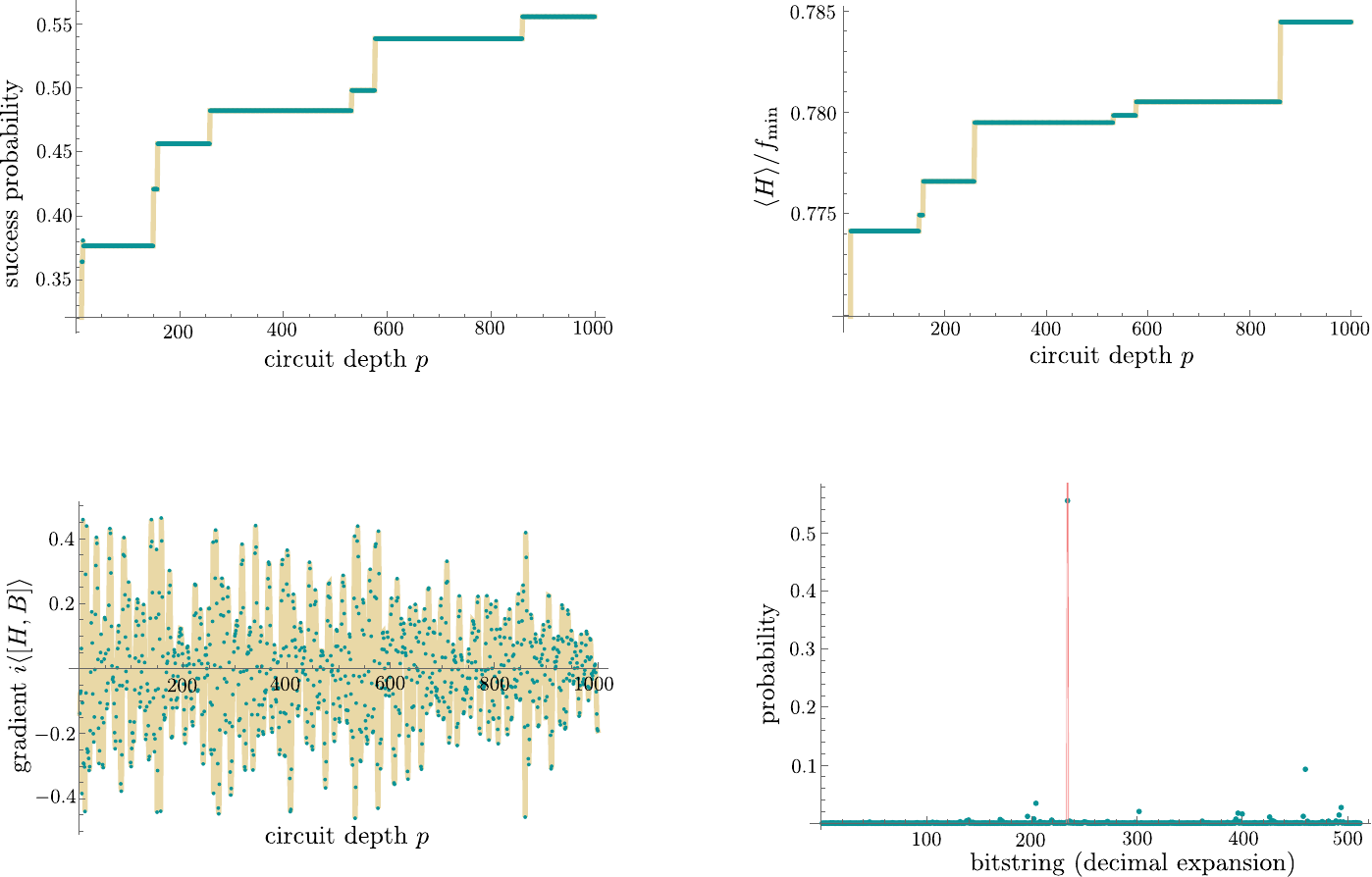}
    \caption{Minimization of a random QUBO function with large step size ($\varepsilon = 1$).
    Both the success probability and the approximation ratio exhibit plateau regions interrupted by sudden jumps to significantly higher values.
    Analogously, the gradient exhibits large and erratic oscillations around zero, effectively overshadowing any initial peak.
    Notably, the final state shows a strong concentration around the global optimum -- only moderately worse than in the case of the tuned step size $\varepsilon = 0.1$.
    }
    \label{figure:BigSteps}
\end{figure*}

\subsection{\label{subsection:StepSizeAnalysis}Step Size Analysis}

For the same random QUBO instance as before, we further investigate the algorithm's behavior for two different step sizes, $\varepsilon \in \{0.01,\, 1\}$.
The smaller step size, $\varepsilon = 0.01$, leads to a very slow traversal of the favorable optimization landscape (see \autoref{figure:SmallSteps}).
The success probability evolves similarly to that in the previously considered QUBO instance with moderate step size;
however, the rate of increase is substantially slower.
As a result, the final outcome distribution is smeared out resembling the behavior seen in the case of an unfavorable landscape.
Due to this slow progress, the approximation ratio does not reach saturation within the first 1000 iterations.
Instead, it exhibits a shallow linear increase over a broad range of steps.
Additionally, after an initial burst, the gradient's amplitude becomes very small -- once again indicating again the slow motion through the optimization landscape.
In summary, the chosen step size clearly is too small to reach the optimum within a reasonable number of iterations.
However, the underlying optimization landscape appears to permit a direct traversal from the initial state to the optimum, provided a sufficiently large circuit depth.

Meanwhile, the larger step size, $\varepsilon = 1$, results in a completely different behavior (see \autoref{figure:BigSteps}).
The success probability remains constant over extended intervals of iterations, exhibiting abrupt increases between these plateaus.
After 1000 iterations, it reaches a relatively high value of $0.55$, which is supported by a strong concentration of the final outcome distribution around the optimal solution.
The approximation ratio follows a similar pattern, remaining constant over the same intervals and showing discrete jumps in between.
We also observe pronounced oscillations with comparatively large amplitude of the gradient -- a behavior that directly results from the large step size.
This aligns with intuition: a large step size is sufficient to approach the optimal solution, but not precise enough to reach it exactly.

\section{\label{section:Conclusion}Outlook and Conclusion}

This work is done under the impression that technological advances of the NISQ era might reach a next stage in a not too far future. 
For this we share the aspirations that quantum computing with deep circuits will become practically feasible.
In this regime the quantum-classical ansatz of variational quantum computing will likely keep its popularity.
Estimating perspectives of practical applicability, spotting new obstacles, and finding promising problem classes is therefore a relevant quest that can already be started today. 
With this work we contribute to the collection of methods, tools, and structural insights that will hopefully lead to a better understanding on why the QAOA practically fails in many examples and when it could in principle work. 

As we have seen, regarding optimization landscapes from a state space perspective allows for a clear analysis that reveals rich but still accessible mathematical structures.
This has to be seen in clear contrast to the, in some extend, more often employed perspective on optimization landscapes in the parameter space of $(\beta,\gamma)\in\mathbb{R}^{2p}$.
The mapping from $(\beta,\gamma)\in\mathbb{R}^{2p}$ does not really respect the natural topology of the problem.
Several, apparently different, local minima and traps in parameter space could for example correspond to one and the same local minimum in state space.
In this sense obstacles that are spotted in the state space picture give us a clear hint towards the persistent geometric core of the underlying problem.
Future work might study which of the geometric properties, we uncovered for the generic case of universal generators, are still present in the non-universal case.
However, as highlighted in \autoref{theorem:UniversalGates}, the optimization problems which do not admit universal generates are merely a null set.

We expect that the analysis of the $\mu$\,-$f$ diagrams, which we introduced as a performance indicator, will turn out as a useful tool for future research. By considering only few basic examples, we merely scratched on the surface of its applicability. Analyzing this indicator and evaluating its impact for practical problem instances will be an essential task for future research.

From the results we have seen so far, it becomes however already clear that there will be no universal applicability of QAOA in deep circuits.
We see that problem instances that avoid unfavorable landscapes must have a specialized underlying structure and are presumable rare.
This is totally in line to the typical observation that no method ever gives a `free lunch'. 

Lastly, we have to point out that many further aspects, that influence the performance of a local search, have been neglected in this work.
This especially includes the actual circuit depth required for approximating the solution of a problem up to a convincing ratio.
For local search routines this depth can be highly problem-specific and might drastically vary with respect to the explicit method in use.
Finding good routines will most likely demand a lot of explicit feature engineering.
Even though we found some promising simple examples in our numerical studies, a conclusive assessment of instances with favorable landscapes can very likely reveal further obstacles.

\subsection*{Acknowledgments}

We thank Daniel Burgarth, Tobias J.\ Osborne, Antonio Rotundo, Bence Marton Temesi, Arne-Christian Voigt, Reinhard F.\ Werner, and Sören Wilkening for helpful discussions.
GK acknowledges financial support by the DAAD and IIT Indore (Kapil Ahuja) for a guest stay. 
LvL and RS acknowledge financial support by the Quantum Valley Lower Saxony.
RS acknowledges financial support by the BMBF project ATIQ.
LB and TZ acknowledge finanical support by the BMBF project QuBRA.

\textbf{Data availability statement:} The depicted data is available upon reasonable request from
the authors.

\twocolumngrid

\bibliographystyle{quantum}
\bibliography{main.bib}

\onecolumngrid

\end{document}